\documentclass[conference]{IEEEtran}
\IEEEoverridecommandlockouts
    
\usepackage[utf8]{inputenc} 
\usepackage[T1]{fontenc}
\usepackage{url}
\usepackage{cite}
\usepackage{longtable}
\usepackage[cmex10]{amsmath} 
\usepackage{authblk}
\usepackage[top=0.74in,bottom=1.04in,left=0.64in,right=0.63in]{geometry}
\usepackage{algorithm}
\usepackage{algorithmic}
\usepackage{amsthm,amssymb,amsfonts} 
\usepackage{algorithmic}
\usepackage{graphicx}
\usepackage{textcomp}
\usepackage{xcolor}
\usepackage{tikz}
\usepackage{ifthen}

\newcommand{\cA}{\mathcal{A}}

\newcommand{\boldA}{\mathbf{A}}
\newcommand{\boldB}{\mathbf{B}}
\newcommand{\boldC}{\mathbf{C}}

\newcommand{\boldS}{\mathbf{S}}
\newcommand{\boldT}{\mathbf{T}}

\newtheorem{pid axiom}{PID Axiom}
\newtheorem{sid axiom}{SID Axiom}

\newtheorem{remark}{Remark}
\newtheorem{definition}{Definition}
\newtheorem{lemma}{Lemma}
\newtheorem{corollary}{Corollary}
\newtheorem{property}{Property}
\newtheorem{theorem}{Theorem}

\newtheorem{proposition}{Proposition}
\usepackage{ifthen}
\newboolean{showAppendix}   
\setboolean{showAppendix}{true}  % true or false

\def\BibTeX{{\rm B\kern-.05em{\sc i\kern-.025em b}\kern-.08em
    T\kern-.1667em\lower.7ex\hbox{E}\kern-.125emX}}
\begin{document}
\title{
Structural Impossibility of Antichain-Lattice \\Partial Information Decomposition
}

\author{\textbf{Aobo Lyu}$^\star$, \textbf{Andrew Clark}$^{\star}$, and \textbf{Netanel Raviv}$^\dagger$\\
$^\star$Department of Electrical and Systems Engineering, Washington University in St. Louis, St. Louis, MO, USA\\
		$^\dagger$Department of Computer Science and Engineering, Washington University in St. Louis, St. Louis, MO, USA\\
   \texttt{aobo.lyu@wustl.edu}, \texttt{andrewclark@wustl.edu}, \texttt{netanel.raviv@wustl.edu}}

\maketitle

\begin{abstract}
Partial Information Decomposition (PID) represents multivariate mutual information via antichain-lattice that aims to specify which source groups can recover which informational components of a target.
For three or more sources, widely desired PID axioms become mutually incompatible. This is often treated as an axiomatic tuning issue. 
This paper argues that the obstruction is representational, rooted in the antichain indexing itself, so that purely axiomatic adjustments within an antichain-lattice structure cannot resolve it in general.
We first introduce System Information Decomposition (SID) for the special target-free three-variable setting, obtaining a self-consistent entropy decomposition with an operational redundancy definition.
More fundamentally, we then show that for general multivariate PID, there is no universal rule that recovers the decomposed mutual information from the antichain-indexed information atoms. 
In particular, two systems can share identical atoms regardless of any axioms while having different mutual~information.
These results reveal the limits of antichain-lattice and motivate relation-based foundations for multivariate information measures.
\end{abstract}
\pagestyle{empty}

\section{Introduction}
Understanding how information is distributed across multiple random variables is central to information theory.
Partial Information Decomposition (PID), introduced by Williams and Beer~\cite{williams2010nonnegative}, provides a framework for addressing this question by decomposing the multivariate mutual information $I(\boldS;T)$ between a set of source variables $\boldS=\{S_1,\dots,S_N\}$ and a target variable $T$ into \emph{information atoms} such as redundant, unique, and synergistic information, indexed by an antichain (redundancy) lattice~\cite{crampton2001completion}.
% This lattice-based PID has proved conceptually powerful and has enabled a wide range of applications~\cite{schneidman2003synergy,hamman2023demystifying}, including quantifying higher-order interactions in neuroscience~\cite{newman2022revealing,varley2023partial} and complex systems~\cite{rosas2020reconciling}, as well as studying privacy and fairness in data disclosure~\cite{rassouli2019data}.
This lattice-based PID has proved conceptually powerful and has enabled a growing range of applications~\cite{hamman2023demystifying}, 
including quantifying neural interactions~\cite{varley2023partial,newman2022revealing},
formalizing causal emergence in complex systems~\cite{rosas2020reconciling,mediano2022greater,yuan2024emergence},
guiding multimodal fusion in machine learning~\cite{liang2023quantifying}.

Despite extensive efforts~\cite{griffith2014intersection, ince2017measuring, harder2013bivariate, bertschinger2014quantifying, lyu2024explicit,bertschinger2013shared}, no existing PID measure simultaneously satisfies all axioms and desired properties. 
A key obstacle is that, for three or more sources, widely desired axioms and properties cannot in general be satisfied simultaneously, as highlighted by several inconsistency and impossibility results in the literature~\cite{bertschinger2013shared,rauh2014reconsidering,kolchinsky2022novel, matthias2025novel, lyu2026multivariate}.
In particular, some results show that the axioms in~\cite{williams2010nonnegative} may violate the intuitive independent identity property~\cite{finn2018pointwise} (see Property~\ref{property: Independent Identity} in Section~\ref{sec:PID}), while the XOR construction in~\cite{rauh2014reconsidering} (see Lemma~\ref{lemma: counter example} in Section~\ref{sec:PID}) reveals that the sum of all atoms may exceed the total information.

% Some research shows the axioms in~\cite{williams2010nonnegative} may violate an intuitive property called {independent identity}~\cite{finn2018pointwise} (see Property~\ref{property: Independent Identity} in Section~\ref{sec:PID}), some shows the axioms may conflict with the inclusion-exclusion principle~\cite{kolchinsky2022novel}. The XOR construction~\cite{rauh2014reconsidering} (see Lemma~\ref{lemma: counter example} in Section~\ref{sec:PID}) reveals that the summation of all atoms may exceed the total information.

Rather than further refining which axioms can or cannot be jointly satisfied, this paper argues that a substantial part of the multivariate PID difficulty is not axiomatic but \emph{representational}:
it is rooted in the lattice itself.
The redundancy antichain-lattice~\cite{crampton2001completion,williams2010nonnegative} is designed to index atoms by which subsets of sources can recover a given informational component about the target.
It naturally encourages a set-theoretic accounting intuition: such patterns can be organized into disjoint atoms whose contributions aggregate in a universal additive manner, often expressed as the \emph{whole-equals-sum-of-parts} (WESP) principle.
However, we show that synergy can can \emph{link} information atoms in ways that the antichain-indexed lattice cannot represent.
This motivates a structural question independent of any particular redundancy formula or axiom set:
\emph{can antichain-indexed atoms universally determine the quantity being decomposed?}
Our main result is negative: the obstruction persists even before choosing axioms—it arises from the limitations of lattice representation capabilities.

Our contributions are as follows.
First, to resolve the multivariate PID inconsistency in a tractable setting and to probe its origin, we introduce the notion of \emph{System Information Decomposition} (SID) for the three-variable case where $T=(S_1,S_2,S_3)$.
In this boundary case, SID provides a compatible axiomatic system with an operational redundancy definition and yields a self-consistent decomposition. More importantly, it shows that higher-order synergy can take a collective form that is not representable by antichain labels alone.
Second and most importantly, we establish a representational impossibility result for general multivariate PID:
for three or more sources, antichain-indexed atoms are not informative enough to determine the decomposed quantity.
In particular, we show that two systems can induce identical antichain-indexed atoms while having different mutual information $I(\boldS;T)$.
Together, these results indicate that the multivariate PID obstruction is not primarily a matter of selecting the ``correct'' axioms but a limitation of the structural representation, motivating alternatives beyond antichain-lattice.

The remainder of the paper is organized as follows.
Section~\ref{sec:PID} reviews PID and recalls a three-source inconsistency.
Section~\ref{sec:SID} presents SID as a self-consistent boundary case and derives its decomposition rules with an operational definition via multivariate \textit{Gacs-Korner} common information.
Section~\ref{sec:limitations} proves the main representational limitation of the lattice via an impossibility theorem and an indistinguishable-pair construction.
Section~\ref{sec:discussion} discusses implications and motivates relation-based foundations for multivariate information measures. 
\ifthenelse{\boolean{showAppendix}}{}{A full version of this paper, including appendices and complete proofs, is available in~\cite{lyu2026structural}.}

\section{Partial Information Decomposition}
\label{sec:PID}

In this section we briefly review the PID of Williams and Beer~\cite{williams2010nonnegative} and recall a three-source inconsistency result. 
\subsection{PID framework and redundancy lattice}
Consider random variables $S_1, S_2, T$ over finite alphabets $\mathcal{S}_1, \mathcal{S}_2, \mathcal{T}$. We denote $S_1$ and $S_2$ as the sources and $T$ as the target. The mutual information~$I(S_1,S_2;T)$ decomposes into redundant, unique, and synergistic atoms (see Figure~\ref{fig:PID}):
{\setlength{\abovedisplayskip}{4pt}
\setlength{\belowdisplayskip}{4pt}
\begin{align}
\label{equ:Information Atoms' relationship_1}
    I(S_1,S_2;T) = \operatorname{Red}(S_1,S_2 \to T) +\operatorname{Syn}(S_1,S_2 \to T) \nonumber\\+\operatorname{Un}(S_1\to T\vert S_2)+ \operatorname{Un}(S_2\to T\vert S_1),
\end{align}
}
where $\operatorname{Red}(S_1,S_2 \to T)$ is redundant information shared by $S_1$ and $S_2$ about $T$, $\operatorname{Un}(S_1 \to T \mid S_2)$ and $\operatorname{Un}(S_2 \to T \mid S_1)$ are unique information from each source, and $\operatorname{Syn}(S_1,S_2 \to T)$ is synergistic information that is only available from the joint observation of $S_1$ and $S_2$.

For each subsystem $(S_1,T)$ and $(S_2,T)$, the atoms satisfy
{\setlength{\abovedisplayskip}{4pt}
\setlength{\belowdisplayskip}{4pt}
\begin{align}\label{equ:Information Atoms' relationship_2}
    I(S_1;T) &= \operatorname{Red}(S_1,S_2 \to T) + \operatorname{Un}(S_1\to T\vert S_2), \mbox{ and}\nonumber\\
    I(S_2;T) &= \operatorname{Red}(S_1,S_2 \to T) + \operatorname{Un}(S_2\to T\vert S_1).
\end{align}}
\begin{figure}[htbp]
\centering
\fbox{\includegraphics[width=.6\linewidth]{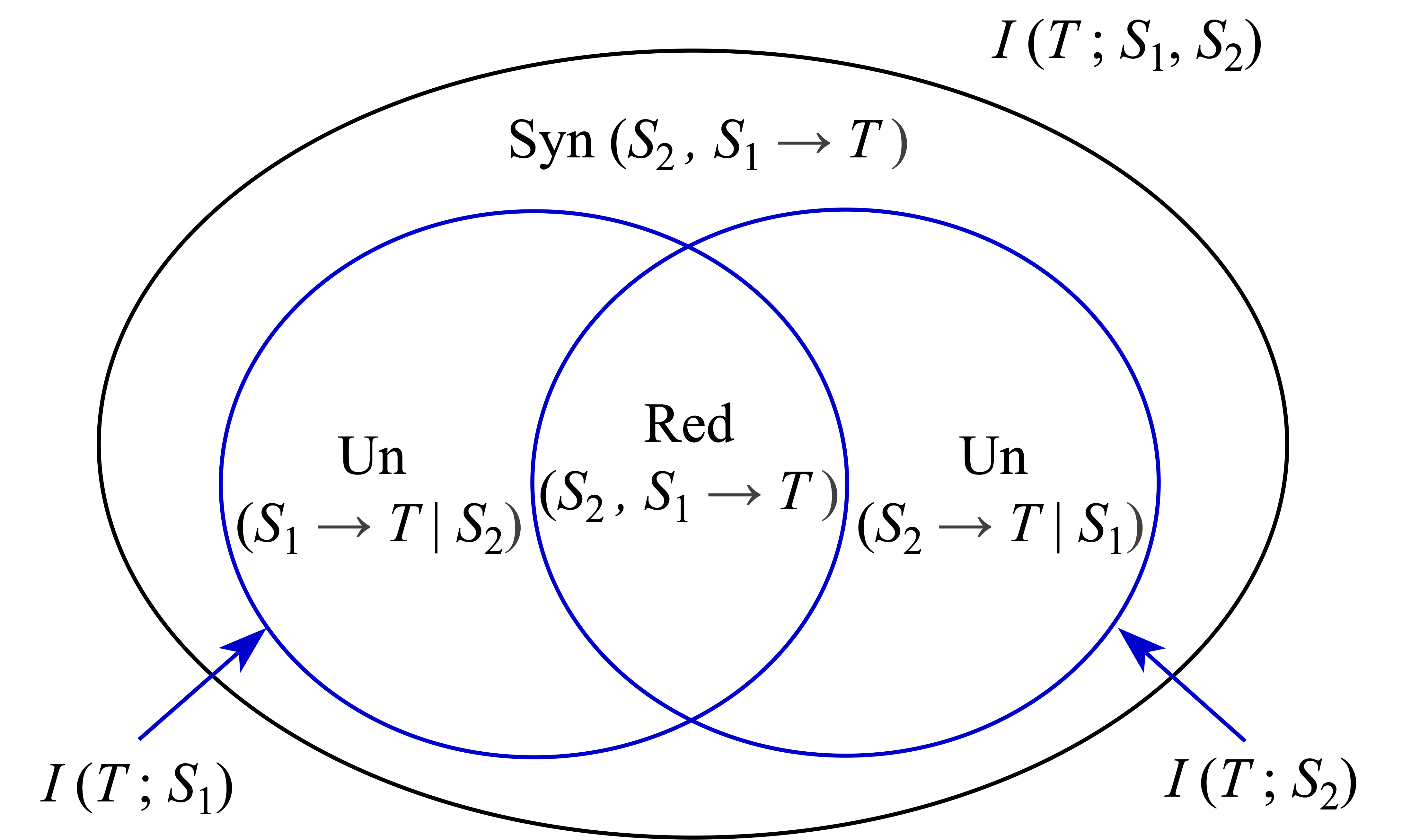 }}
\caption{The structure of PID with two source variables, i.e.,~\eqref{equ:Information Atoms' relationship_1}~\eqref{equ:Information Atoms' relationship_2}.}
\label{fig:PID}
\end{figure}
For general systems with source variables $\boldS = \{S_1,\dots,S_n\}$ and target $T$, PID uses \textit{the redundancy lattice}~$\mathcal{A}(\boldS)$~\cite{williams2010nonnegative,crampton2001completion}, which is the set of antichains formed from the power set of $\boldS$ under set inclusion with a natural order~$\preceq_\boldS$.
\begin{definition}[PID Redundancy Lattice]
\label{definition:PID lattice}
For the set of source variables $\boldS$, the set of antichains is:
{\setlength{\abovedisplayskip}{4pt}
\setlength{\belowdisplayskip}{4pt}
\begin{align*}
\mathcal{A}(\boldS) = \{\alpha \subseteq \mathcal{P}(\boldS)\setminus\{\varnothing\}:\alpha \ne \varnothing, \forall \boldA_i,\boldA_j \in \alpha, \boldA_i \not \subset \boldA_j\},
\end{align*}
}
where $\mathcal{P}(\boldS) $ is powerset of $\boldS$,
and for every~$\alpha,\beta\in\mathcal{A}(\boldS)$, $\beta \preceq_\boldS \alpha$ if for every $\boldA\!\in\!\alpha$ there exists $\boldB\!\in\! \beta$ such that $\boldB\!\subseteq\!\boldA$.
\end{definition}
For ease of exposition, we denote elements of~$\mathcal{A}(\boldS)$ using their indices (e.g., write $\bigl\{\{S_1\}\{S_2\}\bigl\}$ as $\bigl\{\{1\}\{2\}\bigl\}$).
Based on the redundancy lattice, PID assigns a real value to each antichain $\alpha\in\cA(\boldS)$ by a family of functions.

\begin{definition} [Partial Information Decomposition Framework]
\label{pid def:PIDF}
Let~$\boldS$ be a collection of sources and let~$T$ be the target.
A family of functions $\{\Pi^{T}_{\boldA}:\mathcal{A}(\boldA) \rightarrow \mathbb{R}\}_{\boldA \subseteq \boldS}$ is called a family of partial information (PI) functions if it satisfies~PID Axioms~\ref{pid axiom:mutual constrains},~\ref{PID Axiom: Commutativity},~\ref{pid axiom: Monotonicity}, and~\ref{pid axiom: Self-redundancy}, given shortly.
\end{definition}
For simplicity we denote $\Pi^{T}_{i\dots}(\cdot)$ for $\Pi^{T}_{\{S_i\dots\}}(\cdot)$, e.g., $\Pi^{T}_{12}(\{\{1\}\})=\Pi^{T}_{\{S_1,S_2\}}(\{\{S_1\}\})$.
Note that in the case~$\boldS=\{S_1,S_2\}$, the terms in~\eqref{equ:Information Atoms' relationship_1} and~\eqref{equ:Information Atoms' relationship_2} reduce Definition~\ref{pid def:PIDF} to
{\setlength{\abovedisplayskip}{4pt}
\setlength{\belowdisplayskip}{4pt}
\begin{align*}
    \operatorname{Red}(S_1,\!S_2\!\to\!T)\! & =\! \Pi^{T}_{12}(\!\bigl\{\!\{\!1\!\}\!\{\!2\!\}\!\bigl\}\!), \operatorname{Un}(S_1\!\to\! T\vert S_2\!) \!=\! \Pi^{T}_{12}(\!\bigl\{\!\{1\}\!\bigl\}\!),\\
    \operatorname{Syn}(S_1,S_2 \!\to\! T)\! &=\! \Pi^{T}_{12}(\!\bigl\{\!\{12\}\!\bigl\}\!), 
    \operatorname{Un}(S_2\!\to\! T\vert S_1)\! =\! \Pi^{T}_{12}(\!\bigl\{\!\{2\}\!\bigl\}\!).
\end{align*}}
For general systems, each~$\boldA\subseteq \boldS$ and every~$\alpha\in\cA(\boldA)$, the value $\Pi^{T}_{\boldA}(\alpha)$ is called a \textit{PI-atom}.
Intuitively, the PI-atom $\Pi^{T}_{\boldA}(\alpha)$ measures the amount of information provided by each set in the antichain~$\alpha$ to~$T$ and is not attributable to any~$\beta\ne \alpha$ s.t. $\beta \preceq_\boldA \alpha$. To ensure that a PI-function realizes this intended principle, the PID framework imposes a set of structural axioms.
First, it requires the following mutual-information constraints~\cite{williams2010nonnegative} (i.e., the equivalent of~\eqref{equ:Information Atoms' relationship_1} and~\eqref{equ:Information Atoms' relationship_2}).

\begin{pid axiom}[{Whole Equals Sum of Parts}]
\label{pid axiom:mutual constrains}
  For any subsets $\boldA$, $\boldB$ of sources~$\boldS$ with $\boldB\subseteq \boldA$, the sum of PI-atoms decomposed from system $\boldA$ satisfies  
{\setlength{\abovedisplayskip}{4pt}
\setlength{\belowdisplayskip}{4pt}
\begin{align}
\label{equ:PID Information Atoms}
    I(\boldB;T) = \sum_{\beta \preceq_{\boldA} \{\boldB\} } \Pi^{T}_{\boldA}(\beta),
\end{align}
}
where $\{\boldB\}$ is the antichain with a single element~$\boldB$.
\end{pid axiom}
{Equation~\eqref{equ:PID Information Atoms} requires that, for any subsystem $(\boldA,T)$, the mutual information $I(\boldB;T)$ can be recovered by summing the appropriate PI-atoms~\cite{ince2017partial,chicharro2017synergy,rosas2020operational,lizier2013towards}. We refer to this as the \emph{whole-equals-sum-of-parts (WESP)} principle.

Then, PID requires these axioms on the redundancy atom, which further restrict the resulting decomposition.

\begin{pid axiom} [Commutativity]
\label{PID Axiom: Commutativity}
Redundant information is invariant under any permutation $\sigma$ of sources, i.e., $\operatorname{Red}(S_1,\dots,S_N\to T) = \operatorname{Red}(S_{\sigma(1)}, \dots, S_{\sigma(N)}\to T)$.
\end{pid axiom}

\begin{pid axiom} [Monotonicity]
\label{pid axiom: Monotonicity}
Redundant information decreases monotonically as more sources are included, i.e., 
$\operatorname{Red}(S_1,\dots,S_{N},S_{N+1} \to T) \le \operatorname{Red}(S_1,\dots,S_{N} \to T)$.
\end{pid axiom}

\begin{pid axiom} [Self-redundancy]
\label{pid axiom: Self-redundancy}
Redundant information for a single source variable $S_i$ equals the mutual information, i.e., 
$\operatorname{Red}(S_i \to T) = I(S_i;T).$
\end{pid axiom}
Besides, another intuitive property is often considered \cite{ince2017measuring}.
\begin{property}[Independent Identity]
\label{property: Independent Identity}
If $I(S_1;S_2)=0$ and $T=(S_1,S_2)$, then $\operatorname{Red}(S_1,S_2\to T) = 0$.
\end{property}

\subsection{Inconsistency for three or more sources}
An explicit definition for PI-functions for two sources was given in~\cite{lyu2024explicit}.
However, this framework becomes inherently contradictory for three or more source variables, as shown in \cite[Thm.~2]{rauh2014reconsidering}, which we briefly recall below. For completeness, \ifthenelse{\boolean{showAppendix}}{Appendix~\ref{app:counter example}}{Appendix B1 of~\cite{lyu2026structural}}
provides a proof following~\cite{rauh2014reconsidering}.
\begin{lemma} {\cite{rauh2014reconsidering}}
\label{lemma: counter example}
    Let~$S_1$ and~$S_2$ be two independent $\operatorname{Bernoulli}(1/2)$ variables, let~$S_3$ be their exclusive OR (XOR), and let~$T=(S_1,S_2,S_3)$. 
    Then, any PID measure
    satisfies PID Axioms~\ref{PID Axiom: Commutativity}, \ref{pid axiom: Monotonicity}, \ref{pid axiom: Self-redundancy}, Property~\ref{property: Independent Identity},
    violates PID Axiom~\ref{pid axiom:mutual constrains}
    i.e.,
    {\setlength{\abovedisplayskip}{4pt}
\setlength{\belowdisplayskip}{4pt}
\begin{align*}
I(T;\boldS) < \sum_{\beta \preceq_{\boldS} \{\boldS\} } \Pi^{T}_{\boldS}(\beta),
\end{align*}}
where $I(T;\boldS)=2$, but three non-zero atoms have value $1$.
\end{lemma} 
The lattice indexes atoms by source-access patterns, and PID framework imposes an additive accounting rule (Axiom~\ref{pid axiom:mutual constrains}) requiring that each system's mutual information be recovered by summing the atoms in the corresponding down-set, i.e., the WESP principle.
But Lemma~\ref{lemma: counter example} shows that for three sources, the XOR relationship among sources leads to overcounting and violates Axiom~\ref{pid axiom:mutual constrains}~\cite{rauh2014reconsidering}.
Motivated by this obstruction, Section~\ref{sec:SID} introduces \emph{System Information Decomposition} (SID) as a three-variable remedy for the case $T = (S_1,S_2,S_3)$.
Here self-consistency is recovered by modifying the summation rule in~\eqref{equ:PID Information Atoms}, rather than enforcing the WESP additivity.

\section{System Information Decomposition}
\label{sec:SID}
In this section, we consider the three-source case $T=(S_1,S_2,S_3)$, a special boundary case we call \textit{System Information Decomposition} (SID), initially explored in~\cite{lyu2023system}.
Here, the PID of $I(S_1,S_2,S_3;T)$ reduces to a decomposition of the joint entropy $H(S_1,S_2,S_3)$.
To avoid the overcounting in Section~\ref{sec:PID}, we replace Axiom~\ref{pid axiom:mutual constrains} by a modified summation rule over a subset of atoms.
We use the following lattice.

\begin{definition}[SID Half Lattice]
\label{definition:SID_half_lattice}
For $\boldS = \{S_1,S_2,S_3\}$, let
{\setlength{\abovedisplayskip}{4pt}
\setlength{\belowdisplayskip}{4pt}
\begin{align}
\label{equation:sid_lattice}
\mathcal{A}^*(\boldS) =& \{\alpha \in \mathcal{A}(\boldS): \exists \boldA_k \in \alpha, |\boldA_k|=1\}, \\
=&\big\{
\,\{\{1\}\{2\}\{3\}\},\{\{1\}\{2\}\},\{\{1\}\{3\}\},\{\{2\}\{3\}\},\nonumber \\
\{\{1\}\{23&\}\},\{\{2\}\{13\}\}, \{\{3\}\{12\}\},\{\{1\}\},\{\{2\}\},\{\{3\}\}\, \big\}.  \nonumber
\end{align}}
where $\preceq_\boldS$ is as in Definition~\ref{definition:PID lattice}.
\end{definition}

Intuitively, $\mathcal{A}^*(\boldS)$ removes antichains that contain no singleton sources.
When $T = \boldS$, these singleton-free patterns do not appear in the chain-rule expansions of $H(\boldS)$, and hence are not needed in this setting.

\begin{definition} [System Information Decomposition Framework]
\label{sid def:SIDF}
A family of functions 
 $\{\Psi_{\boldA} :\mathcal{A^*}(\boldA) \rightarrow \mathbb{R}\}_{\boldA \subseteq \boldS}$ 
is called a family of system information (SI) functions if it satisfies~SID Axioms~\ref{sid axiom:mutual constrains},~\ref{sid axiom: commutativity},~\ref{sid axiom: Monotonicity}, and~\ref{sid axiom: Self-redundancy}, given shortly. 
\end{definition}
For every~$\boldA\subseteq\boldS$ and every~$\alpha\in\cA^*(\boldA)$, the value $\Psi_{\boldA}(\alpha)$ is called a SI-atom. Our aim is to measure the information contributed by every subset in~$\alpha$ to the whole system $\boldA$, which is not already accounted for by any antichain~$\beta \preceq_\boldA \alpha$.

The SID half lattice can be understood as a refinement of the PID redundancy lattice for three sources (Definition~\ref{definition:PID lattice}), by removing all antichains that do not contain any singleton source (see Figure~\ref{fig:SID_3PID}(B)).
See 
\ifthenelse{\boolean{showAppendix}}{Appendix~\ref{app:compairson}}{Appendix~A of~\cite{lyu2026structural}}
for further comparison between SID and two-source PID.

\begin{figure}[htbp]
\centering
% \vspace*{2pt}
\fbox{\includegraphics[width=0.97\linewidth]{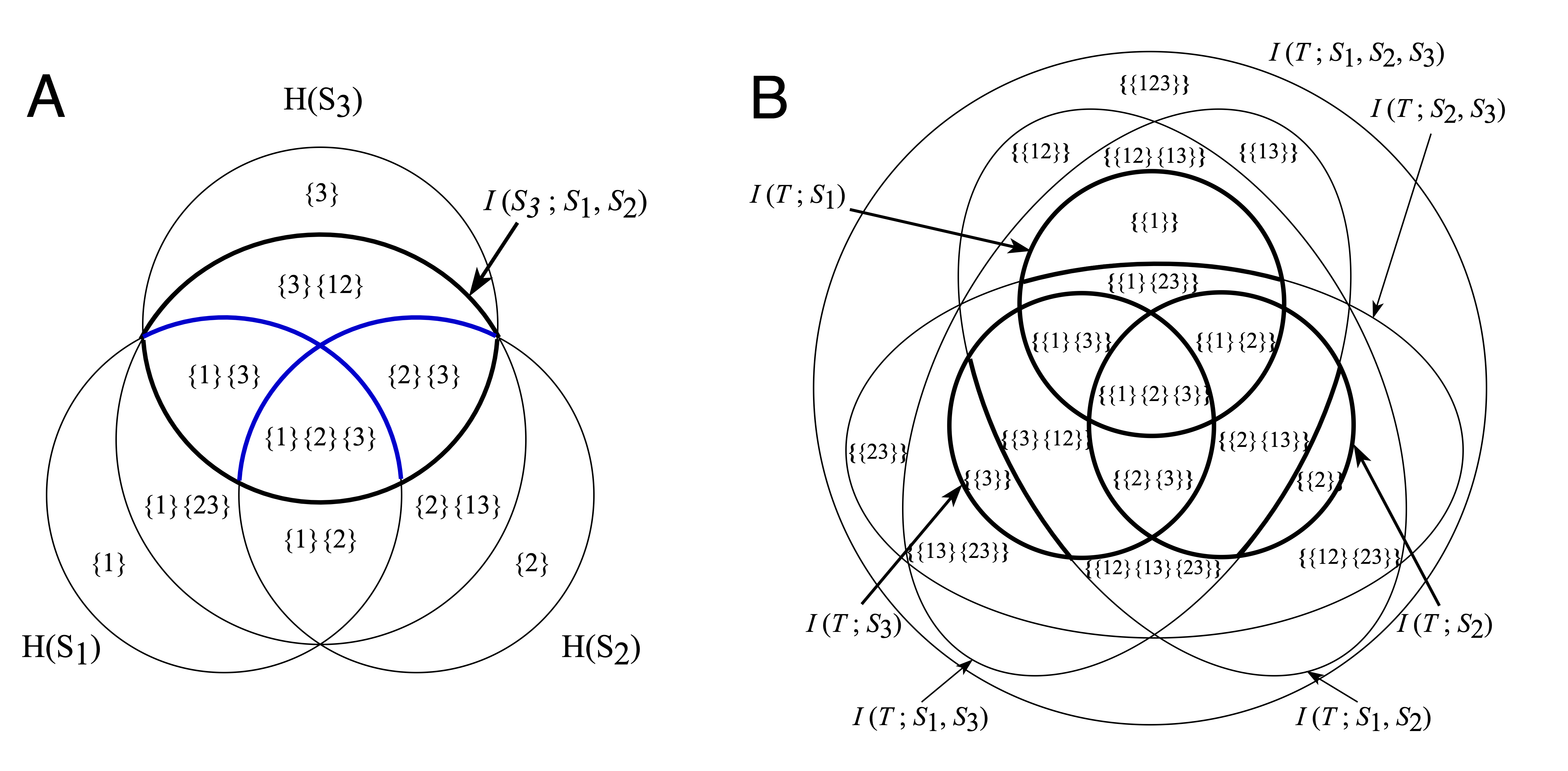 }}
\caption{Comparison between SID and three-source PID. (A) Three-variable SID.
(B) Three-source PID, where the antichains in bold contain at least one singleton source, whose structure is consistent with SID.}
\label{fig:SID_3PID}
\end{figure}
We retain commutativity and monotonicity in analogous form, and adapt self-redundancy to the setting.
PID Axiom~\ref{pid axiom:mutual constrains}, which leads to the inconsistency demonstrated in Lemma~\ref{lemma: counter example}, will be modified shortly.
Similar to PID, we define SID redundant information
as $\operatorname{Red}(S_1,S_2,S_3)=\Psi_{\boldS}(\{\{S_1\}\{S_2\}\{S_3\}\})$, and for all distinct $i,j$ in $\boldS$, let $\operatorname{Red}(S_i,S_j)=\Psi_{\{S_i,S_j\}}(\{\{S_i\}\{S_j\}\})$.

\setcounter{sid axiom}{1}
\begin{sid axiom}[Commutativity]
\label{sid axiom: commutativity}
SID redundant information is invariant under any permutation $\sigma$ of sources, i.e., $    \operatorname{Red}(S_1,S_2,S_3) = \operatorname{Red}(S_{\sigma(1)},S_{\sigma(2)},S_{\sigma(3)})$.
\end{sid axiom}

\begin{sid axiom} [Monotonicity]
\label{sid axiom: Monotonicity}
SID redundant information decreases monotonically as more sources are included, i.e., 
$\operatorname{Red}(S_1,S_2,S_3) \le \min_{i,j\in[3]}\{\operatorname{Red}(S_i,S_j)\}$.
\end{sid axiom}

\begin{sid axiom} [Self-redundancy]
\label{sid axiom: Self-redundancy}
SID redundant information for two variables $S_i,S_j$ equals the mutual information, i.e. $\operatorname{Red}(S_i,S_j) = I(S_i;S_j)$.

\end{sid axiom}

We then revisit PID Axiom~\ref{pid axiom:mutual constrains} and propose the following alternative axiom; see
\ifthenelse{\boolean{showAppendix}}{Appendix~\ref{app:derive SID}}{Appendix~B2 of~\cite{lyu2026structural}}
for the derivation.

\setcounter{sid axiom}{0}
\begin{sid axiom}
\label{sid axiom:mutual constrains}
For any set of variables $\boldS=\{S_1,S_2,S_3\}$ and $\boldB\subseteq\boldA\subseteq\boldS$ with $|\boldB| \le 2$,
the entropy of $\boldB$ is decomposed~as
\begin{align}
H(\boldB)=\sum_{\alpha \in \mathcal{A}^*(\boldA):\, \alpha \preceq_{\boldA} \{\boldB\}} \Psi_{\boldA}(\alpha),    
\end{align}
and when~$|\boldB|=|\boldS|=3$ we have for all distinct $i,j,k \in[3]$,
\begin{align}
\label{equ:SID Information Atoms}
H(\boldS) = \sum_{\alpha \in \mathcal{A}^*(\boldS)} \Psi_{\boldS}(\alpha) - \Psi_{\boldS}(\{\{ij\}\{k\}\}).
\end{align}
\end{sid axiom}

\begin{proposition}[Symmetric synergy from SID Axiom~\ref{sid axiom:mutual constrains}] 
\label{prop:sid_symmetry}
Under SID Axiom~\ref{sid axiom:mutual constrains}, the three pair-to-single SI-atoms coincide:
\begin{align}
\Psi_{\boldS}(\{\{12\}\{3\}\})
=
\Psi_{\boldS}(\{\{13\}\{2\}\})
=
\Psi_{\boldS}(\{\{23\}\{1\}\}).
\label{eq:sid_correction_equal}
\end{align}
\end{proposition}

\begin{proof}
Apply SID Axiom~\ref{sid axiom:mutual constrains} to $(i,j,k)=(1,2,3)$ and its permutations; the exclusion term is permutation-invariant.
\end{proof}

\begin{remark} 
Proposition~\ref{prop:sid_symmetry} shows that the exclusion term in
SID Axiom~\ref{sid axiom:mutual constrains} is permutation-invariant, i.e., the three atoms encode the same symmetric contribution. Consequently, SID does not treat all atoms as universally disjoint parts, and self-consistency excludes exactly one copy of this linked term.
\end{remark}
For the XOR system in Lemma~\ref{lemma: counter example} with $T{=}(S_1,S_2,S_3)$, this yields for any distinct $i,j,k\in[3]$
\[
H(T)=\Psi_{\boldS}(\{\{i\}\{jk\}\})+\Psi_{\boldS}(\{\{j\}\{ik\}\})=2,
\]
where zero-valued SI-atoms are omitted.
This linkage can be accounted for explicitly in the boundary case $|\boldS|=3$ via~\eqref{equ:SID Information Atoms}, but, as shown in the next section, it cannot be captured by antichain-indexed atoms in general multivariate systems.

The following lemma shows that only the redundancy atom needs to be defined; the remaining atoms are then uniquely determined via linear constraints implied by Axiom~\ref{sid axiom:mutual constrains}.
A proof is provided in \ifthenelse{\boolean{showAppendix}}{Appendix~\ref{app:UniqueUpToOneAtom}}{Appendix~B3 of~\cite{lyu2026structural}}
alongside explicit definitions for all SI atoms given~$\operatorname{Red}(S_1,S_2,S_3)$.

\begin{lemma}\label{lemma:UniqueUpToOneAtom}
Let \(\boldS = \{S_1, S_2, S_3\}\) be a three-variable system in the SID framework, and
\(\Psi_{123}(\cdot)\) denote its SI-atoms. 
Then, once the value of any one SI-atom is fixed, the values of all remaining SI-atoms are uniquely determined by SID Axiom~\ref{sid axiom:mutual constrains}.
\end{lemma}

Therefore, any definition of $\operatorname{Red}(S_1,S_2,S_3)$ implies unique definitions of all SI atoms that automatically satisfy SID Axiom~\ref{sid axiom:mutual constrains}.
To satisfy SID Axioms~\ref{sid axiom: commutativity}, \ref{sid axiom: Monotonicity} and~\ref{sid axiom: Self-redundancy}, we adopt a multivariate form of the Gács-Körner common information\footnote{The Gács-Körner common information is defined as $\operatorname{CI}(S_1, S_2) \triangleq \max_{Q} H(Q), \text{s.t. } H(Q | S_1) = H(Q | S_2) = 0$.}~\cite{gacs1973common}(see, e.g.,~\cite{tyagi2011function}) as the redundancy measure.

\begin{definition}[Operational Definition of Redundancy] 
\label{definition:red}
For system \( S_1, S_2, S_3 \), the redundant information is defined as $\operatorname{Red}(S_i,S_j)\triangleq I(S_i;S_j)$ for all distinct $i , j\in \{1,2,3\}$, and
{\setlength{\abovedisplayskip}{3.5pt}
\setlength{\belowdisplayskip}{3.5pt}
\begin{align*}
\operatorname{Red}(S_1,S_2,S_3)&\triangleq \max_{Q}\{H(Q)\mid\!H(Q|S_i)=0, \forall i \!\in \![3]\},
\end{align*} }
where the maximization is taken over all variables \( Q \) defined over the Cartesian product of the alphabets of \( S_1, S_2, S_3 \).
\end{definition}

Gács-Körner common information was also used in~\cite{griffith2014intersection} to define redundancy in a PID-related context.
The following lemma is proved in 
\ifthenelse{\boolean{showAppendix}}{Appendix~\ref{app: satisfaction}.}{Appendix~B4 of~\cite{lyu2026structural}.}

\begin{lemma}
\label{lemma: satisfaction}
Definition~\ref{definition:red} satisfies SID~Axioms~\ref{sid axiom:mutual constrains},~\ref{sid axiom: commutativity}, \ref{sid axiom: Monotonicity}, and~\ref{sid axiom: Self-redundancy}.
\end{lemma}
Section~\ref{sec:SID} shows that in the three-variable setting, one can restore consistency by replacing a universal WESP-type summation with the modified entropy rule in SID Axiom~\ref{sid axiom:mutual constrains}.
Importantly, Proposition~\ref{prop:sid_symmetry} already highlights a representational gap: while the antichain-lattice indexes atoms, global accounting may require additional relations among atoms that are not encoded by the antichain itself.
In SID, this missing relation can be supplied explicitly as the symmetric correction term in~\eqref{equ:SID Information Atoms}, but for general multivariate systems, such extra structure cannot be captured by antichain-indexed atoms.

Motivated by this, Section~\ref{sec:limitations} investigates whether the absence of explicit relations among information atoms constitutes a fundamental obstruction to antichain-lattice-based multivariate information decomposition, i.e., whether the antichain-indexed atoms can determine the decomposed quantity in a universal way, in particular $I(\boldS;T)$.

\section{Structural Limitations of Antichain-Lattice}
\label{sec:limitations}
Existing PID approaches typically begin with the antichain lattice and then posit axioms for antichain-indexed information atoms, seeking PI-functions that satisfy those axioms.
In this section, our goal is to evaluate whether the antichain lattice itself is capable of representing and decomposing information.

The approach is as follows. We consider a restricted class of distributions, which we denote as \emph{the antichain-realizable atom model}, such that the values of all antichain-indexed atoms can be derived from an intuitive first principle.
We then construct two multivariate systems belonging to this restricted class, and prove that these random variables have the same atoms for each antichain $\alpha$, and yet different mutual information $I(\boldS;T)$. Hence, this proves that there is no way of defining atoms such that mutual information can be reliably computed from the  atom values alone regardless of any axiom system, or equivalently, that antichain-lattice-based atoms are inadequate for decomposing mutual information.

Recall the standard PID setup. 
Definition~\ref{definition:PID lattice} fixes the antichain lattice $\mathcal{A}(\boldS)$ (ordered by $\preceq_{\boldS}$) as the index set for information atoms.
Each $\alpha\in\mathcal{A}(\boldS)$ is an antichain of source sets (subsets of $[n]$), and the intended principle is:
\begin{remark}[Intuitive First Principle]
In Definition~\ref{pid def:PIDF}, the atom labeled by $\alpha$ is intended to capture the information about $T$ that is recoverable from each source group $\boldB\in\alpha$,
but not already recoverable under any strictly weaker label $\beta\prec_{\boldS}\alpha$.    
\end{remark}
For example, when $\boldS=\{S_1,S_2,S_3\}$ and $\alpha=\{\{1\}\{23\}\}$, the atom labeled by $\alpha$ is meant to capture information about $T$ that one can obtain either from $S_1$ alone or from $(S_2,S_3)$ jointly, but not from $S_2$ or $S_3$ alone.

Based on this intuitive first principle, we focus on a restricted class of distributions constructed from latent components.
In this class, each latent component is designed to be recoverable from exactly the source groups prescribed by one lattice label, so the corresponding atom values are fixed by construction.
We formalize this idealized setting next.
\begin{definition}[Antichain-realizable atom model]
\label{def:ideal-semantics}
Fix random variables $x_1,\ldots,x_m$ and index sets
$J_1,\ldots,J_n,J_T\subseteq[m]$ with $J_T\subseteq \bigcup_{i\in[n]} J_i$.
Define $T\triangleq (x_j)_{j\in J_T}$ and $\boldS=\{S_1,\ldots,S_n\}$ by $S_i \triangleq (x_j)_{j\in J_i}$ for all $i \in[n]$.

We say that $(\boldS,T)$ admits an \emph{antichain-realizable atom model} if (i) for each $j\in[m]$, $H(x_j|T)=0$ implies $j\in J_T$;
(ii) for each $i\in[n]$, the variables $\{x_j:j\in J_i\}$ are mutually independent; and
(iii) for every $j\in J_T$ and every $\boldB\subseteq[n]$, 
 writing $S_{\boldB}\triangleq (S_i)_{i\in\boldB}$,
{\setlength{\abovedisplayskip}{3.5pt}
\setlength{\belowdisplayskip}{3.5pt}
\begin{align*}
\text{either }\ H(x_j\mid S_{\boldB})=0\quad
 \text{or}\quad
I(x_j;S_{\boldB})=0.     
\end{align*}
}
\end{definition}
Definition~\ref{def:ideal-semantics} restricts attention to a very narrow class of constructed distributions, but this is sufficient to obtain the \emph{counterexample proof} we need.
More importantly, in this class, 
the lattice’s \emph{intuitive first principle} uniquely induces the values of antichain-indexed atoms, without invoking any redundancy formula or axiom system.
The next lemma makes this correspondence explicit (proved in 
\ifthenelse{\boolean{showAppendix}}{Appendix~\ref{app:proof_principal_antichain}}{Appendix~C of~\cite{lyu2026structural}}).
\begin{lemma}
\label{lem:canonical-atom-assignment}
Assume $(\boldS,T)$ satisfies Definition~\ref{def:ideal-semantics}.
For each $j\in J_T$ with $H(x_j)>0$, define its recovering sets
{\setlength{\abovedisplayskip}{3.5pt}
\setlength{\belowdisplayskip}{3.5pt}
\begin{align}
    \label{eq:recovering}
    \mathsf{Rec}(x_j)\triangleq\{\boldB\subseteq[n]: H(x_j\mid S_{\boldB})=0\},
\end{align}
}
and additionally, define its corresponding antichain as the set of minimal recovering sets
{\setlength{\abovedisplayskip}{3.5pt}
\setlength{\belowdisplayskip}{3.5pt}
\[
\alpha(x_j)\triangleq \Bigl\{\boldB\in\mathsf{Rec}(x_j):\forall\,\boldC\subsetneq \boldB,\ \boldC\notin\mathsf{Rec}(x_j)\Bigr\}.
\]}
Then, for each $\alpha\in\mathcal{A}(\boldS)$, we have
\[
\Pi^T_{\boldS}(\alpha)=H(U_\alpha), \text{ where } U_\alpha\triangleq \bigl(x_j:\ j\in J_T,\ \alpha(x_j)=\alpha \bigr).
\]
\end{lemma}
Lemma~\ref{lem:canonical-atom-assignment} yields a ``ground-truth'' antichain-indexed atoms
$\bigl(\Pi^T_{\boldS}(\alpha)\bigr)_{\alpha\in\mathcal{A}(\boldS)}$.
In particular, the values $\Pi^T_{\boldS}(\alpha)$ do not depend on any auxiliary redundancy definition or axiom choices.
For instance, the XOR construction in Lemma~\ref{lemma: counter example} satisfies Definition~\ref{def:ideal-semantics}, yielding three atoms labeled by $\{\!\{1\}\!\{23\}\!\},\{\!\{2\}\!\{13\}\!\}$ and $\{\!\{3\}\!\{12\}\!\}$ have value~$1$, consistent with~\cite{rauh2014reconsidering}.

From now on we restrict attention to joint distributions that satisfy
Definition~\ref{def:ideal-semantics}, where the antichain-indexed atoms
$\bigl(\Pi^{T}_{\boldS}(\beta)\bigr)_{\beta\in\mathcal{A}(\boldS)}$
are fixed by construction and do not depend on any redundancy formula or axiom choices.

Then, we focus on a crucial problem: the notion of a \emph{decomposition} of $I(\boldS;T)$ into antichain-indexed atoms presupposes 
that $I(\boldS;T)$ be a function of all atoms.
However, the next theorem shows that no such universal reconstruction is possible, even in this idealized setting.
\begin{theorem}
\label{thm:no-functional-ideal}
Let $K$ be the number of antichains in $\mathcal{A}(\boldS)$, where $|\boldS|\ge 3$,
there is no function $f:\mathbb{R}^{K}\to\mathbb{R}$ such that
{\setlength{\abovedisplayskip}{3.5pt}
\setlength{\belowdisplayskip}{3.5pt}
\begin{align}
\label{eq:functional-relation}
I(\boldS;T)
=
f\bigl( (\Pi^{T}_{\boldS}(\beta))_{\beta\in\mathcal{A}(\boldS)} \bigr)
\end{align}}
for all joint distributions $(\boldS,T)$ that satisfy Definition~\ref{def:ideal-semantics}.
\end{theorem}
We prove Theorem~\ref{thm:no-functional-ideal} by exhibiting two joint distributions that satisfy
Definition~\ref{def:ideal-semantics} with identical atoms
$\bigl(\Pi^T_{\boldS}(\alpha)\bigr)_{\alpha\in\mathcal{A}(\boldS)}$, yet different values of $I(\boldS;T)$.
This rules out any universal reconstruction function of the form~\eqref{eq:functional-relation}.

We now consider two three-source systems
$(\hat{S}_1,\hat{S}_2,\hat{S}_3,\hat{T})$ and $(\tilde{S}_1,\tilde{S}_2,\tilde{S}_3,\tilde{T})$,
depicted in Fig.~\ref{fig:System23}. 
Both systems are constructed from latent Boolean variables $x_1,\dots,x_9$.

In $(\hat{\boldS},\hat{T})$, let $x_1,x_2,x_4,x_5,x_7,x_8\sim\mathrm{Bernoulli}(1/2)$ be mutually independent and let
$
x_3=x_1\oplus x_2, x_6=x_4\oplus x_5, x_9=x_7\oplus x_8.
$
Then, we set
$\hat{S}_1=(x_1,x_4,x_7)$, 
$\hat{S}_2=(x_2,x_5,x_8)$, 
$\hat{S}_3=(x_3,x_6,x_9)$, and
$\hat{T}=(x_1,x_5,x_9)$.

In $(\tilde{\boldS},\tilde{T})$, let $x_1,x_2,x_4,x_5,x_7\sim\mathrm{Bernoulli}(1/2)$ be mutually independent and let
$
x_3=x_1\oplus x_2, x_6=x_4\oplus x_5, x_9=x_1\oplus x_5, x_8=x_7\oplus x_1\oplus x_5,
$
so that $x_9=x_7\oplus x_8=x_1\oplus x_5$ holds by construction.
Then, we set
$\tilde{S}_1=(x_1,x_4,x_7)$, 
$\tilde{S}_2=(x_2,x_5,x_8)$, 
$\tilde{S}_3=(x_3,x_6,x_9)$, and
$\tilde{T}=(x_1,x_5,x_9)$.
\begin{figure}[htbp]
    \centering
    \fbox{\includegraphics[width=.9\linewidth]{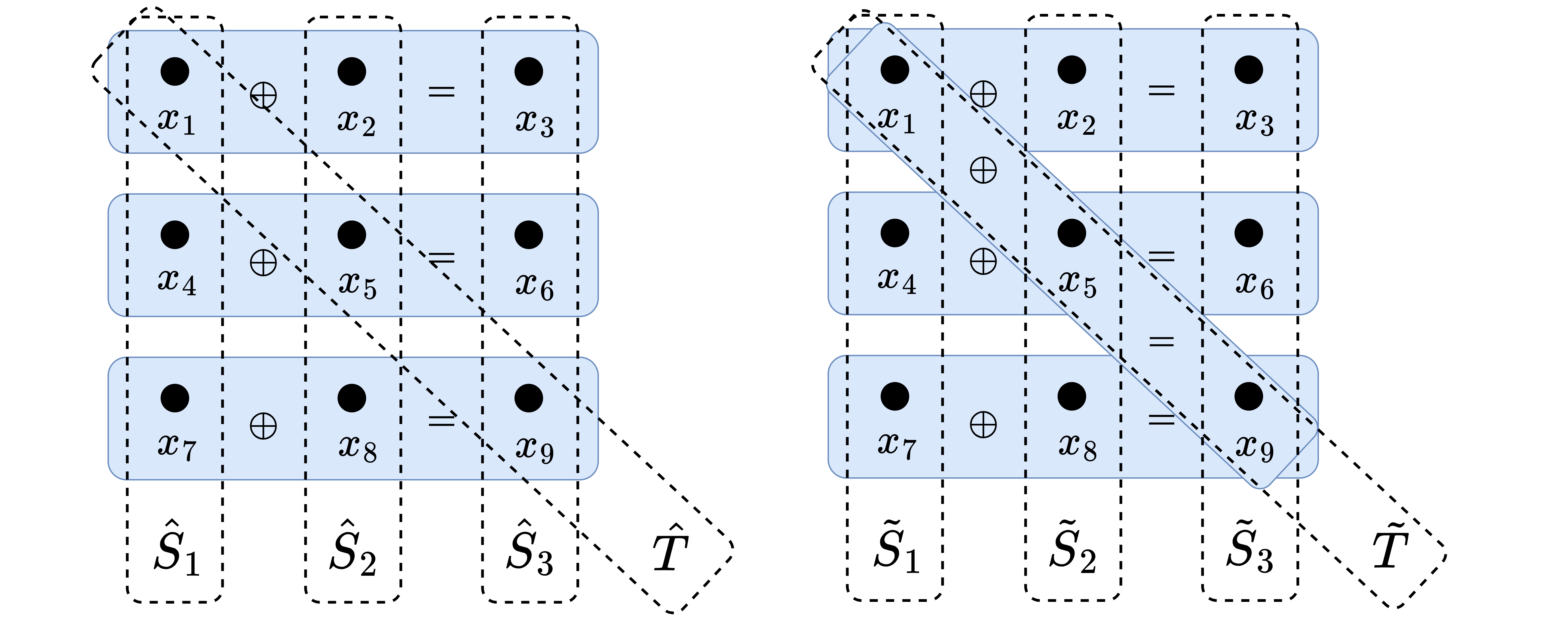}}
    \caption{Three-source systems $(\hat{S}_1,\hat{S}_2,\hat{S}_3,\hat{T})$ and $(\tilde{S}_1,\tilde{S}_2,\tilde{S}_3,\tilde{T})$ constructed from latent bits $x_1$ to $x_9$.}
    \label{fig:System23}
\end{figure}

The system $(\tilde{S},\tilde{T})$ enforces an additional global constraint (equivalently, one fewer latent degree of freedom),
which changes the joint dependence structure and hence the value of $I(\boldS;T)$, while leaving the resulting atoms under Definition~\ref{def:ideal-semantics} unchanged. We formalize this in the following lemma, which proved in 
\ifthenelse{\boolean{showAppendix}}{Appendix~\ref{app:pair-construction}.}{Appendix~D of~\cite{lyu2026structural}.}
% Appendix~\ref{app:pair-construction}.
Explicit probability tables for both systems are provided in \ifthenelse{\boolean{showAppendix}}{Appendix~\ref{app:alltargets}}{Appendix~E of~\cite{lyu2026structural}}.
\begin{lemma}[Witness pair]
\label{lemma:groundtruth-pair}
The systems ($\hat{S},\hat{T}$) and ($\tilde{S},\tilde{T}$) described above satisfy:
\begin{enumerate}
\item their atoms coincide,
$\Pi^{\hat{T}}_{\hat{\boldS}}(\beta)=\Pi^{\tilde{T}}_{\tilde{\boldS}}(\beta), \forall \beta\in\mathcal{A}(\boldS),$
(the atoms indexed by $\{\{1\}\{23\}\}$, $\{\{2\}\{13\}\}$, and $\{\{3\}\{12\}\}$ are $1$, the rest are $0$); and
\item their mutual informations differ: $I(\hat{\boldS};\hat{T})\neq I(\tilde{\boldS};\tilde{T})$.
\end{enumerate}
\end{lemma}
Lemma~\ref{lemma:groundtruth-pair} implies Theorem~\ref{thm:no-functional-ideal} immediately. 
Indeed, if~\eqref{eq:functional-relation} held for some universal $f$, then we get the contradiction
{\setlength{\abovedisplayskip}{3.5pt}
\setlength{\belowdisplayskip}{3.5pt}
\[
I(\hat{\boldS};\hat{T})
\!=\! f\!\left(\!\bigl(\Pi^{\hat{T}}_{\hat{\boldS}}(\beta)\bigr)_{\!\beta\in\mathcal{A}(\hat{\boldS})}\!\right)
\!=\! f\!\left(\!\bigl(\Pi^{\tilde{T}}_{\tilde{\boldS}}(\beta)\bigr)_{\!\beta\in\mathcal{A}(\tilde{\boldS})}\!\right)
\!=\! I(\tilde{\boldS};\tilde{T}).
\]}
The counterexample extends to any $n>3$ by adjoining extra sources that are independent of the current variables in both systems, leaving the atoms and mutual information unchanged.

In summary, we exhibited two systems with identical atoms
$\bigl(\Pi^{T}_{\boldS}(\alpha)\bigr)_{\alpha\in\mathcal{A}(\boldS)}$ but different values of $I(\boldS;T)$.
Therefore, even in the case where the lattice meaning is realized exactly,
$I(\boldS;T)$ is not uniquely determined by the atoms.
This rules out any universal reconstruction map from atoms to $I(\boldS;T)$, and hence rules out
any multivariate information decomposition that relies solely on the antichain-lattice.

\section{Discussion}
\label{sec:discussion}
This work argues that the difficulty of PID is not primarily an issue of axiom selection or redundancy tuning, but a representational limitation of the antichain lattice itself.
As a boundary case, Section~\ref{sec:SID} introduced System Information Decomposition (SID) for three variables.
By replacing WESP with a modified summation rule on a reduced lattice, SID restores self-consistency and shows that higher-order synergy can act as a symmetric collective contribution.
The appearance of a symmetric correction term exposes the core obstruction: correct global accounting may require relations among atoms that are not specified by antichain labels alone.

Section~\ref{sec:limitations} formalizes this obstruction in an idealized setting.
Even when the lattice meaning is realized exactly (via a ground-truth construction), antichain-indexed atoms do not universally determine the decomposed quantity $I(\boldS;T)$,
since they do not encode the cross-atom constraints (Proposition~\ref{prop:sid_symmetry}) or relations among target components (e.g., $\tilde{T}$ in Lemma~\ref{lemma:groundtruth-pair}). 
Consequently, the quantity $I(\boldS;T)$ can vary while the atoms remain unchanged. 
This does not preclude the existence of useful multivariate decompositions, but it indicates that additional structure beyond antichain lattice is essential.

A natural direction is therefore to augment atoms with explicit relations—for example, relation-based representations such as hypergraphs that encode global constraints or higher-order dependencies directly—while retaining operational meaning and computability.

\newpage
\IEEEtriggeratref{14}
%Bibliography
\bibliographystyle{unsrt}
\bibliography{references}

\ifthenelse{\boolean{showAppendix}}{
\newpage
\appendix
\subsection{Comparison between SID and two sources PID}
\label{app:compairson}

\begin{figure}[htbp]
\centering
\fbox{\includegraphics[width=0.97\linewidth]{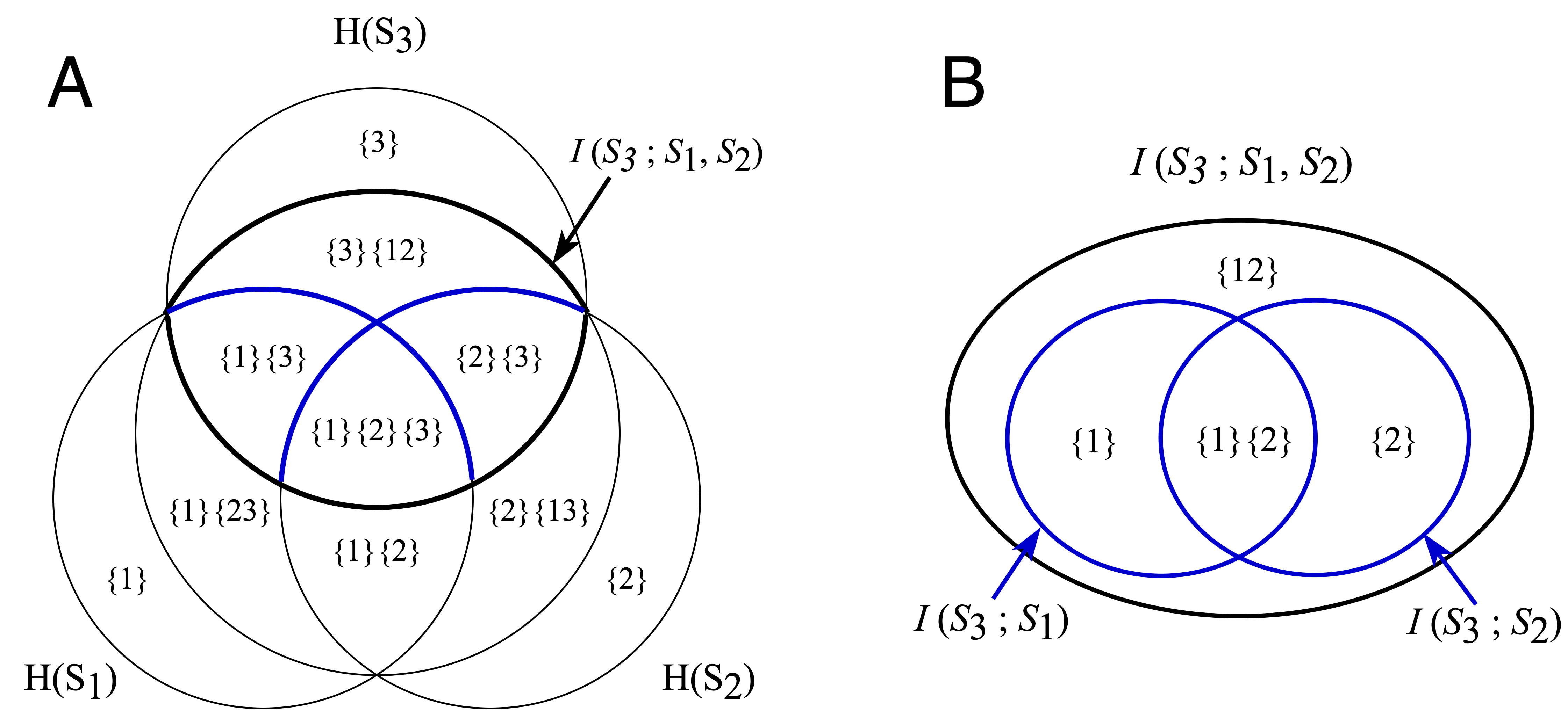 }}
\caption{Comparison between SID and two-source PID.}
\label{fig:SID_2PID}
\end{figure}

SID extends the scope of 2-source PID from mutual information $I(\boldS\setminus S_i;S_i)$ to the joint entropy $H(\boldS)$ of the system. 
In SID (target-free), each SI-atom represents information that a certain combination of variables provides redundantly to the system as a whole.
% Unlike PID, SID is a target-free framework, where each atom represents the information redundantly provided by the sets with the antichain. In 
For instance, in Fig.~\ref{fig:SID_2PID}(A), the SI-atom $\Psi_{123}(\{\{3\}\{12\}\})$ represents information in $S_3$ that is also contributed synergistically by $S_1$ and $S_2$. This directly corresponds to the PI-atom $\Pi_{12}^3(\{\{12\}\})$ in the PID view (Fig.~\ref{fig:SID_2PID}(B)), where we have a target $T=S_3$ and sources $S_1,S_2$.

\subsection{Proofs of Main Results}
To prove the lemmas in the paper, we first need the following lemma and corollary.
%\red{[Subsystem consistency as Lemma~1, with the following two equations as corollaries.]}

Axiom~\ref{pid axiom:mutual constrains} couples decompositions obtained from different subsystems, as captured by the following lemma.}
\begin{lemma}[{Subsystem Consistency}]
\label{lemma: subsystem consistency}
For system with sources $\boldS$ and target $T$ and any $\boldA,\boldB,\boldC\subseteq \boldS$ such that~$\boldA\subseteq\boldC\cap \boldB$, let $\Pi_\boldC^T,\Pi_\boldB^T$ (see Definition~\ref{pid def:PIDF}) which satisfy PID Axiom~\ref{pid axiom:mutual constrains}. Then we have that
    \begin{align}
    \label{equ:subsystem}
        \sum_{\beta \preceq_\boldC \{\boldA\}} \Pi^{T}_{\boldC}(\beta)=\sum_{\beta \preceq_\boldB \{\boldA\}} \Pi^{T}_{\boldB}(\beta).
    \end{align}    
\end{lemma}
\begin{proof}
    Apply PID Axiom~\ref{pid axiom:mutual constrains} with~$\boldA\subseteq\boldB\subseteq\boldS$ and then with~$\boldA\subseteq\boldC\subseteq\boldS$.
    %The result is directly from~\ref{equ:PID Information Atoms}, by applying PID Axiom~\ref{pid axiom:mutual constrains} on $I(\boldA;T)$ with $\boldA\subseteq\boldB$ and $\boldA\subseteq\boldC$.
\end{proof}
Intuitively, Lemma~\ref{lemma: subsystem consistency} states that the total information that the subset $\boldA$ provides about $T$ is \emph{independent} of the subsystem in which it is computed.
%; the decomposition is consistent under marginalization of sources.}
To illustrate this concept, consider the system in Figure~\ref{fig:PID}.
For the atoms decomposed from the system $(S_1,T)$, the quantity $\Pi^{T}_{1}(\bigl\{\{1\}\bigl\})$ reflects the (redundant) information that $S_1$ provides about~$T$.
If we add a source~$S_2$ to this system, this information will be further decomposed into the redundant information $\Pi^{T}_{12}(\bigl\{\{1\}\{2\}\bigl\})$ from $S_1,S_2$, and the unique information $\Pi^{T}_{12}(\bigl\{\{1\}\bigl\})$ only from $S_1$ but not $S_2$. 
Below are three axioms regarding the redundant information $\operatorname{Red}(S_1,\dots,S_N\to T)$---which is reflected by the PI-atom $\Pi^{T}_{\boldS}(\bigl\{\{1\}\dots\{N\}\bigl\})$---for any multivariate system~$\boldS$.

\begin{corollary}
\label{corollary: two result}
For the system $(S_1, S_2,S_3,T)$ and its sub-system $(S_1,\!S_2,\!T)$, $(S_1,\!T)$, the decomposed PI-atoms from different sub-systems have the following relationship:
\begin{align}
\label{equ:cross scale_0}
\Pi^{T}_{1}(\!\bigl\{\!\{1\}\!\bigl\}\!) = \Pi^{T}_{12}(\!\bigl\{\!\{1\}\{2\}\!\bigl\}\!) +\Pi^{T}_{12}(\!\bigl\{\!\{1\}\!\bigl\}\!),
\end{align}  
similarly, for the system $(S_1,S_2,S_3,T)$ and $(S_1,S_2,T)$,
\begin{align}
\label{equ:cross scale}
\Pi^{T}_{12}(\!\bigl\{\!\{1\}\{2\}\!\bigl\}\!) &= \Pi^{T}_{123}(\!\bigl\{\!\{1\}\!\{2\}\!\{3\}\!\bigl\}\!) +\Pi^{T}_{123}(\!\bigl\{\!\{1\}\!\{2\}\!\bigl\}\!),\\
\Pi^{T}_{12}(\!\bigl\{\!\{1\}\!\bigl\}\!) &= \Pi^{T}_{123}(\!\bigl\{\!\{1\}\{3\}\!\bigl\}\!)+\Pi^{T}_{123}(\!\bigl\{\!\{1\}\{23\}\!\bigl\}\!)\nonumber \\&+\Pi^{T}_{123}(\!\bigl\{\!\{1\}\!\bigl\}\!).\label{equ:cross scale_2}
\end{align}
\end{corollary}
\begin{proof}
For the system $(S_1, S_2, T)$ and $(S_1, T)$, according to Lemma~\ref{lemma: subsystem consistency}, let $\boldA=\{S_1,S_2\}$, $\boldB=\{S_1\}$, and $\boldC=\{S_1\}$, then we have
\begin{align}\label{equ:decompose_of_0}
\Pi^{T}_{1}(\!\bigl\{\!\{1\}\!\bigl\}\!) = \Pi^{T}_{12}(\!\bigl\{\!\{1\}\{2\}\!\bigl\}\!) +\Pi^{T}_{12}(\!\bigl\{\!\{1\}\!\bigl\}\!).
\end{align}  
Similarly, for the system $(S_1, S_2, T)$ and $(S_2, T)$, we have 
\begin{align}
\Pi^{T}_{2}(\!\bigl\{\!\{2\}\!\bigl\}\!) = \Pi^{T}_{12}(\!\bigl\{\!\{1\}\{2\}\!\bigl\}\!) +\Pi^{T}_{12}(\!\bigl\{\!\{2\}\!\bigl\}\!),\nonumber
\end{align}  
where the information atoms contained in both $\Pi^{T}_{1}(\!\bigl\{\!\{1\}\!\bigl\}\!)$ and $\Pi^{T}_{2}(\!\bigl\{\!\{2\}\!\bigl\}\!)$ is $\Pi^{T}_{12}(\!\bigl\{\!\{1\}\{2\}\!\bigl\}\!)$.

Then, following the same approach, we focus on the system $(S_1,S_2,S_3,T)$ and $(S_1, T)$, i.e., we let $\boldA=\{S_1,S_2,S_3\}$, $\boldB=\{S_1\}$, and $\boldC=\{S_1\}$.
Then, by~Lemma~\ref{lemma: subsystem consistency} we have
\begin{align}\label{equ:decompose_of_1}
\Pi^{T}_{1}&(\!\bigl\{\!\{1\}\!\bigl\}\!) =\Pi^{T}_{123}(\!\bigl\{\!\{1\}\{2\}\{3\}\!\bigl\}\!) +\Pi^{T}_{123}(\!\bigl\{\!\{1\}\{2\}\!\bigl\}\!)\nonumber\\
&+\Pi^{T}_{123}(\!\bigl\{\!\{1\}\{3\}\!\bigl\}\!)+\Pi^{T}_{123}(\!\bigl\{\!\{1\}\{23\}\!\bigl\}\!)+\Pi^{T}_{123}(\!\bigl\{\!\{1\}\!\bigl\}\!).
\end{align}  
Similarly, for the system $(S_1, S_2,S_3,T)$ and $(S_2, T)$, we have
\begin{align}
\Pi^{T}_{2}&(\!\bigl\{\!\{2\}\!\bigl\}\!) =\Pi^{T}_{123}(\!\bigl\{\!\{1\}\{2\}\{3\}\!\bigl\}\!) +\Pi^{T}_{123}(\!\bigl\{\!\{1\}\{2\}\!\bigl\}\!)\nonumber\\
&+\Pi^{T}_{123}(\!\bigl\{\!\{2\}\{3\}\!\bigl\}\!)+\Pi^{T}_{123}(\!\bigl\{\!\{2\}\{13\}\!\bigl\}\!)+\Pi^{T}_{123}(\!\bigl\{\!\{2\}\!\bigl\}\!),\nonumber
\end{align}  
where the information atoms contained in both $\Pi^{T}_{1}(\!\bigl\{\!\{1\}\!\bigl\}\!)$ and $\Pi^{T}_{2}(\!\bigl\{\!\{2\}\!\bigl\}\!)$ are $\Pi^{T}_{123}(\!\bigl\{\!\{1\}\{2\}\{3\}\!\bigl\}\!) $ and $\Pi^{T}_{123}(\!\bigl\{\!\{1\}\{2\}\!\bigl\}\!)$.
Hence, we have
\begin{align}
\Pi^{T}_{12}(\!\bigl\{\!\{1\}\{2\}\!\bigl\}\!) &= \Pi^{T}_{123}(\!\bigl\{\!\{1\}\{2\}\{3\}\!\bigl\}\!) +\Pi^{T}_{123}(\!\bigl\{\!\{1\}\{2\}\!\bigl\}\!),\nonumber
\end{align}  
where $\{\Pi^{T}_{12}(\!\bigl\{\!\{1\}\{2\}\!\bigl\}\!)\}$ and $\{\Pi^{T}_{123}(\!\bigl\{\!\{1\}\{2\}\{3\}\!\bigl\}\!)$, $ \Pi^{T}_{123}(\!\bigl\{\!\{1\}\{2\}\!\bigl\}\!)\}$ are the only atom(s) that are contained in both $I(S_1,T)$ (i.e., $\Pi^{T}_{1}(\!\bigl\{\!\{1\}\!\bigl\}\!)$) and $I(S_2,T)$ (i.e., $\Pi^{T}_{2}(\!\bigl\{\!\{2\}\!\bigl\}\!)$) from the decompositions under the scope of $(S_1, S_2,T)$ and $(S_1, S_2,S_3,T)$.
Therefore,~\eqref{equ:cross scale} is proved. 
%\red{[Explain (in words) why there are no more atoms in the above equation.]}

Then, by~\eqref{equ:cross scale}, \eqref{equ:decompose_of_0}, and \eqref{equ:decompose_of_1}, we have
\begin{align*}
\Pi^{T}_{12}(\!\bigl\{\!\{1\}\!\bigl\}\!) \!=\! \Pi^{T}_{123}(\!\bigl\{\!\{1\}\!\{3\}\!\bigl\}\!)\!+\!\Pi^{T}_{123}(\!\bigl\{\!\{1\}\!\{23\}\!\bigl\}\!)\!+\!\Pi^{T}_{123}(\!\bigl\{\!\{1\}\!\bigl\}\!),
\end{align*}
which is~\eqref{equ:cross scale_2}.
\end{proof}
Axiom~\ref{pid axiom: Monotonicity} also implies another lemma, as follows.

\begin{lemma} [Nonnegativity]
\label{Lemma: Nonnegativity}
Partial Information Decomposition satisfies $\operatorname{Red}(S_1,\dots,S_N \to T) \ge 0$.
\begin{proof}
Add a constant variable $S^*$
% , where $H(S^*)=0$, 
to the sources and obtain $\operatorname{Red}(\boldA \to T) \ge \operatorname{Red}(\boldA,S^* \to T)$, which is $0$ since the constant variable $S^*$ cannot provide any information to the target $T$. 
\end{proof}
\end{lemma}

Using Lemma~\ref{Lemma: Nonnegativity} and Corollary~\ref{corollary: two result}, we prove Lemma~\ref{lemma: counter example}, \ref{lemma:UniqueUpToOneAtom}, and~\ref{lemma: satisfaction} sequentially.
\subsubsection{Proof of Lemma~\ref{lemma: counter example}}
\label{app:counter example}
\begin{proof}
In $(\bar{S}_1,\bar{S}_2,\bar{S}_3,\bar{T})$, let $\bar{S}_1$ and $\bar{S}_2$ be two independent~$\text{Bernoulli}(1/2)$ variables, let $\bar{S}_3 = \bar{S}_1 \oplus \bar{S}_2$, and let~$\bar{T}=(\bar{S}_1,\bar{S}_2,\bar{S}_3)$. Therefore, we have
\begin{align}
\label{equ:I(T;S_1,S_2,S_3)}
I(\bar{T};\bar{S}_1,\bar{S}_2,\bar{S}_3)=2.
\end{align}

Our subsequent proof idea is to use Property~\ref{property: Independent Identity} to obtain the values of all PI-atoms in any system with two sources and the target variable (i.e. $(\bar{S}_1,\bar{S}_2,\bar{T}),(\bar{S}_1,\bar{S}_3,\bar{T})$ and $(\bar{S}_2,\bar{S}_3,\bar{T})$) and then show that their sum will be greater than the joint mutual information of the system $(\bar{S}_1,\bar{S}_2,\bar{S}_3,\bar{T})$. For simplicity, throughout the following proof, we adopt the convention that all statements are considered for distinct~$i,j,k \in \{1,2,3\}$. 

Firstly, by Property~\ref{property: Independent Identity} (Independent Identity), and since $\bar{T} =(\bar{S}_1,\bar{S}_2,\bar{S}_3)\overset{\text{det}}{=}(\bar{S}_i,\bar{S}_j)$ we have
\begin{align}
\label{equ:all_three_is_zero}
     \Pi^{\bar{T}}_{ij}(\bigl\{\{i\}\{j\}\bigl\}) = 0.
\end{align}
Considering that
    $\Pi^{\bar{T}}_{ij}(\bigl\{\{i\}\{j\}\bigl\})= \Pi^{\bar{T}}_{123}(\bigl\{\{1\}\{2\}\{3\}\bigl\}) + \Pi^{\bar{T}}_{123}(\bigl\{\{i\}\{j\}\bigl\})$,
which is identical to~\eqref{equ:cross scale}, and by Axiom~\ref{pid axiom: Monotonicity} (Monotonicity) and Lemma~\ref{Lemma: Nonnegativity} (Nonnegativity) we have 
\begin{align}
\label{equ:pi12=0}
    \Pi^{\bar{T}}_{123}(\bigl\{\{1\}\{2\}\{3\}\bigl\}) = \Pi^{\bar{T}}_{123}(\bigl\{\{i\}\{j\}\bigl\}) =  0.
\end{align}
Similarly, \eqref{equ:Information Atoms' relationship_2} implies that $I(\bar{T};\bar{S}_i)= \Pi^{\bar{T}}_{ij}(\bigl\{\{i\}\{j\}\bigl\}) + \Pi^{\bar{T}}_{ij}(\bigl\{\{i\}\bigl\})$, and since $ I(\bar{T};\bar{S}_i) =1$ and due to~\eqref{equ:all_three_is_zero}, it follows that $\Pi^{\bar{T}}_{ij}(\bigl\{\{i\}\bigl\})=1$, 
which by Corollary~\ref{corollary: two result} (specifically~\eqref{equ:cross scale_2}), equals
\begin{align}
\label{equ:pi-red}
    % \eqref{equ:un=1} &=
    \Pi^{\bar{T}}_{123}(\bigl\{\{i\}\bigl\}) + \Pi^{\bar{T}}_{123}(\bigl\{\{i\}\{jk\}\bigl\}) + \Pi^{\bar{T}}_{123}(\bigl\{\{i\}\{k\}\bigl\}).
\end{align}
Then, by \eqref{equ:pi12=0} and \eqref{equ:pi-red}, we have
\begin{align}
\label{equ:sun=1}
\Pi^{\bar{T}}_{123}(\bigl\{\{i\}\bigl\})+ \Pi^{\bar{T}}_{123}(\bigl\{\{i\}\{jk\}\bigl\}) = 1,
\end{align}
and hence,
\begin{align*}
I(\bar{T};\bar{S}_1,\bar{S}_2,\bar{S}_3) &\ge \Pi^{\bar{T}}_{123}(\bigl\{\{1\}\bigl\}) + \Pi^{\bar{T}}_{123}(\bigl\{\{1\}\{23\}\bigl\})  \\&+ \Pi^{\bar{T}}_{123}(\bigl\{\{2\}\bigl\})\nonumber+ \Pi^{\bar{T}}_{123}(\bigl\{\{2\}\{13\}\bigl\}) \\&+ \Pi^{\bar{T}}_{123}(\bigl\{\{3\}\bigl\}) + \Pi^{\bar{T}}_{123}(\bigl\{\{3\}\{12\}\bigl\}) =3,
\end{align*}
which contradicts~\eqref{equ:I(T;S_1,S_2,S_3)}.
\end{proof}

\subsubsection{{Derivation} of SID Axiom~\ref{sid axiom:mutual constrains}}
\label{app:derive SID}
In SID, the mutual information between any two variables and the third one can be decomposed similarly to two-source PID. 
That is, for any distinct $i,j,k \in\{1,2,3\}$, $I({S_i,S_j}; S_k)$ splits into four SI-atoms (analogous to~\eqref{equ:Information Atoms' relationship_1}):
\begin{align}
I(S_i,S_j&;S_k) = \;\Psi_{\boldS}(\{\{i\}\{j\}\{k\}\}) + \Psi_{\boldS}(\{\{i\}\{k\}\}) \nonumber \\
&\phantom{=}+ \Psi_{\boldS}(\{\{j\}\{k\}\}) + \Psi_{\boldS}(\{\{ij\}\{k\}\}),
\label{equ:two_one_mutul}
\end{align}
and the two-variable mutual information $I(S_i; S_k)$ corresponds to two of those atoms (analogous to~\eqref{equ:Information Atoms' relationship_2}):
\begin{align}
\label{equ:one_one_mutul_decompose}
I(S_i;S_k) = \Psi_{\boldS}&(\{\{i\}\{j\}\{k\}\}) + \Psi_{\boldS}(\{\{i\}\{k\}\}).
\end{align}
Recall that we have $H(S_k) = I(S_i,S_j;S_k) + H(S_k|S_i,S_j)$ for any $k\in [3]$, and $\Psi_{\boldS}(\{\{k\}\})$ represents the information provided by~$S_k$ alone, i.e., $ \Psi_{\boldS}(\{\{k\}\})=H(S_k \mid S_i,S_j)$.
Therefore, we have
\begin{align}
\label{eq:single_variable_entropy}
H(S_k) &= \, I(S_i,S_j; S_k) + H(S_k|S_i,S_j)\nonumber \\
&\overset{\eqref{equ:two_one_mutul}}{=}  \Psi_{\boldS}(\{\{i\}\{j\}\{k\}\})+ \Psi_{\boldS}(\{\{ij\}\{k\}\}) \nonumber \\
+& \Psi_{\boldS}(\{\{j\}\{k\}\}) + \Psi_{\boldS}(\{\{i\}\{k\}\})+\Psi_{\boldS}(\{\{k\}\}) \nonumber \\
&=\sum_{\beta \preceq_{\boldS} \{\{S_k\}\}}\Psi_{\boldS}(\beta).
\end{align}

Similarly, 
for any two variables \(\{S_i,S_k\}\subseteq\boldS\), by combining $H(S_k|S_i)=H(S_k)-I(S_i;S_k)$ 
with~\eqref{equ:one_one_mutul_decompose} and~\eqref{eq:single_variable_entropy}, we have
\begin{align*}
H(S_k|S_i) &= \Psi_{\boldS}(\{\{ij\}\{k\}\}) + \Psi_{\boldS}(\{\{j\}\{k\}\}) + \Psi_{\boldS}(\{\{k\}\}),
\end{align*}
which, combined  with the fact that $H(S_i,S_k)=H(S_i)+H(S_k|S_i)$ and with~\eqref{eq:single_variable_entropy}, shows that the joint entropy of any two variables is the sum of all atoms dominated by that pair:
\begin{align}
\label{eq:two_variables_entropy}
H&(S_i,S_k) = \Psi_{\boldS}(\{\{i\}\{j\}\{k\}\}) + \Psi_{\boldS}(\{\{i\}\{k\}\}) \nonumber \\
&+ \Psi_{\boldS}(\{\{i\}\{j\}\}) + \Psi_{\boldS}(\{\{jk\}\{i\}\})+\Psi_{\boldS}(\{\{i\}\}) \nonumber \\ 
&+\Psi_{\boldS}(\{\{ij\}\{k\}\}) + \Psi_{\boldS}(\{\{j\}\{k\}\}) + \Psi_{\boldS}(\{\{k\}\}) \nonumber \\
&= \Sigma_{\{i,k\}},
\end{align}
where~$\Sigma_{\{i,k\}}$ is the summation of all atoms corresponding to antichains that are dominated either by~$\{\{S_i\}\}$ or by~$\{\{S_k\}\}$.

However, when extending the decomposition to the joint entropy of all three variables, the SID framework deviates from WESP due to the presence of synergy-induced redundancy. 
This discrepancy can be directly demonstrated as follows. 
By combining the fact that $H(S_i,S_j,S_k)=H(S_i,S_k)+H(S_j|S_i,S_k)$ with $\Psi_{\boldS}(\{\{j\}\})=H(S_j|S_i,S_k)$ and~\eqref{eq:two_variables_entropy},
\begin{align}
\label{eq:three_variables_entropy}
H(S_i,S_j,S_k)&=\Sigma_{\{i,k\}} + \Psi_{\boldS}(\{\{j\}\})\nonumber \\
&=\Sigma - \Psi_{\boldS}(\{\{ik\}\{j\}\}),
\end{align}
where~$\Sigma$ is the summation of all~$10$ atoms~$\Psi_{\boldS}(\alpha),\alpha\in\cA^*(\boldS)$.
Thus, unlike PID Axiom~\ref{pid axiom:mutual constrains}, we find that the total entropy is less than the sum of its decomposed parts by exactly $\Psi_{\boldS}(\{\{ij\}\{k\}\})$. 
In other words, WESP does not hold in SID due to this necessary exclusion.

Motivated by~\eqref{eq:single_variable_entropy}, \eqref{eq:two_variables_entropy}, and~\eqref{eq:three_variables_entropy}, 
we propose the SID Axiom~\ref{sid axiom:mutual constrains} alternative to PID Axiom~\ref{pid axiom:mutual constrains}.

\subsubsection{Proof of Lemma~\ref{lemma:UniqueUpToOneAtom}}
\label{app:UniqueUpToOneAtom}
\begin{proof}
We consider the linear constraints relating to the following ten unknowns (the ten SI-atoms of a three-variable system). Define the following vector of atoms:
\begin{align*}
X = \Bigl[
&\,\Psi_{123}(\{\{1\}\{2\}\{3\}\}),\\
&\Psi_{123}(\{\{1\}\{2\}\}),\Psi_{123}(\{\{1\}\{3\}\}),\Psi_{123}(\{\{2\}\{3\}\}),\\
&\Psi_{123}(\{\{1\}\{23\}\}),\Psi_{123}(\{\{2\}\{13\}\}),\Psi_{123}(\{\{3\}\{12\}\}),\\
&\Psi_{123}(\{\{1\}\}),\,\Psi_{123}(\{\{2\}\}),\,\Psi_{123}(\{\{3\}\}) \Bigr]^T,
\end{align*}
and the following vector of entropies:
\begin{align*}
Y \;=\; \Bigl[
&\;H(S_1),\;H(S_2),\;H(S_3),\\
&\,H(S_1,S_2),\;H(S_1,S_3),\;H(S_2,S_3),\\
&\,H(S_1,S_2,S_3),\;H(S_1,S_2,S_3),\;H(S_1,S_2,S_3)
\Bigr]^T.
\end{align*}
Then, the nine constraints which arise from SID Axiom~\ref{sid axiom:mutual constrains}, along with the conditions from SID Axiom~\ref{sid axiom:mutual constrains} are as follows.
\begin{align*}
\begin{bmatrix}
1 & 1 & 1 & 0 & 1 & 0 & 0 & 1 & 0 & 0\\
1 & 1 & 0 & 1 & 0 & 1 & 0 & 0 & 1 & 0\\
1 & 0 & 1 & 1 & 0 & 0 & 1 & 0 & 0 & 1\\
1 & 1 & 1 & 1 & 1 & 1 & 0 & 1 & 1 & 0\\
1 & 1 & 1 & 1 & 1 & 0 & 1 & 1 & 0 & 1\\
1 & 1 & 1 & 1 & 0 & 1 & 1 & 0 & 1 & 1\\
1 & 1 & 1 & 1 & 1 & 1 & 0 & 1 & 1 & 1\\
1 & 1 & 1 & 1 & 1 & 0 & 1 & 1 & 1 & 1\\
1 & 1 & 1 & 1 & 0 & 1 & 1 & 1 & 1 & 1
\end{bmatrix}
X \;=\; Y.
\end{align*}

Solving the system provides the following definition of all SI atoms given~$\operatorname{Red}(S_1,S_2,S_3)$:
\begin{align}\label{equation:explcitSIatomsGivenRED}
    \Psi_{123}(\{1\}\{2\}\{3\})&\triangleq\operatorname{Red}(S_{1},S_2,S_{3})\nonumber\\
    \Psi_{123}(\big\{\{ij\}\big\})&=H(S_i)+H(S_j)\nonumber\\
    &-H(S_i,S_j)-\operatorname{Red}(S_{1},S_2,S_{3})\nonumber\\
    \Psi_{123}(\big\{\{i\}\{jk\}\big\})&=-H(S_1)-H(S_2)-H(S_3)\nonumber\\
    &+H(S_1,S_2)+H(S_1,S_3)+H(S_2,S_3)\nonumber\\
    &-H(S_1,S_2,S_3)+\operatorname{Red}(S_{1},S_2,S_{3}) \nonumber \\
    \Psi_{123}(\big\{\{i\}\big\})&=H(S_1,S_2,S_3)-H(S_j,S_k) 
\end{align}
for all~$i,j,k$ such that \(\{i,j,k\} = \{1,2,3\}\).
\end{proof}

% These nine equations form a linear system of rank \(9\) in the ten unknowns. Consequently, the solution space is one-dimensional: any single SI-atom can serve as a free variable, and once its value is specified, the other nine become uniquely determined. This proves that fixing the value of exactly one SI-atom is sufficient to determine all others.
% \red{[Incomplete proof. The solution space is of the form~$\{ t\Vec{v}+\Vec{u}|t\in\mathbb{R} \}$, hence fixing any entry~$i$ of a vector in the solution space \textbf{where~$v_i\ne 0$} determines~$t$, and hence determines the remaining entries of that vector as well. 
% Why would~$\operatorname{Red}$ be the nonzero entry? Solution: Solve the system, and show that the first entry of the vector~$\Vec{v}$ is nonzero.]}
% \end{proof}

% \subsection{Explicit Forms of SI atoms given \(\operatorname{Red}\).}
% \label{app:atoms}
% Building on Lemma~\ref{lemma:UniqueUpToOneAtom}, if we further specify a particular value for \(\operatorname{Red}(S_1, S_2, S_3)\) (i.e., for \(\Psi_{123}(\{\{1\}\{2\}\{3\}\})\)), then each remaining atom can be written in closed form as follows.
% \begin{align*}
%     \Psi_{123}(\{1\}\{2\}\{3\})&=\operatorname{Red}(S_{1},S_2,S_{3})\\
%     \Psi_{123}(\big\{\{ij\}\big\})&=H(S_i)+H(S_j)\\&-H(S_i,S_j)-\operatorname{Red}(S_{1},S_2,S_{3})\\
%     \Psi_{123}(\big\{\{i\}\{jk\}\big\})&=-H(S_1)-H(S_2)-H(S_3)\\&+H(S_1,S_2)+H(S_1,S_3)+H(S_2,S_3)\\&-H(S_1,S_2,S_3)+\operatorname{Red}(S_{1},S_2,S_{3}) \\\Psi_{123}(\big\{\{i\}\big\})&=H(S_1,S_2,S_3)-H(S_j,S_k) \qedhere
% \end{align*}
% where \(\{i,j,k\} = \{1,2,3\}\).

\subsubsection{Proof of Lemma~\ref{lemma: satisfaction}}
\label{app: satisfaction}
\begin{proof}~
SID Axiom~\ref{sid axiom: commutativity} (Commutativity) is clearly satisfied by Definition~\ref{definition:red}, since the condition is symmetric with respect to the input variables; SID Axiom~\ref{sid axiom: Self-redundancy} (Self-redundancy) is also satisfied by the definition. %~\eqref{equ:def_0}
SID Axiom~\ref{sid axiom: Monotonicity} (Monotonicity) follows from Definition~\ref{definition:red} since adding a new variable imposes additional constraints on the maximization:
%, leading to a non-increasing sequence of values:
\begin{align*}
    \operatorname{Red}&(S_1,S_2,S_3) = \max_{Q}\{H(Q) : H(Q \mid S_i) = 0, \forall i \in [3] \} \\ 
    &\leq \max_{Q} \{H(Q) : H(Q \mid S_i) = 0, H(Q \mid S_j) = 0 \} \\ 
    &= \operatorname{CI}(S_i, S_j), %\quad \forall i \neq j \in \{1,2,3\},
\end{align*}
for every distinct~$i$ and~$j$ in~$\{1,2,3\}$, where the last equality follows from the definition $\operatorname{CI}(S_1, S_2) \triangleq \max_{Q} H(Q), \text{s.t. } H(Q | S_1) = H(Q | S_2) = 0$~\cite{gacs1973common}. Moreover, since \(\operatorname{CI}(S_i, S_j) \leq I(S_i; S_j)\) \cite{gacs1973common},
% \red{[provide exact reference for this claim.]}
it follows that
\begin{align*}
    \operatorname{Red}(S_1, S_2, S_3) &\leq I(S_i; S_j), %\forall i \neq j \in \{1,2,3\}.\qedhere
\end{align*}
for every distinct~$i$,~$j$ in~$\{1,2,3\}$, hence SID Axiom~\ref{sid axiom: Monotonicity} follows.
\end{proof}

\subsection{Proof of Lemma~\ref{lem:canonical-atom-assignment}}
\label{app:proof_principal_antichain}

\begin{proof}
Fix $j\in J_T$ with $H(x_j)>0$ and recall $S_{\boldB}\triangleq (S_i)_{i\in\boldB}$.
We first present two basic properties.

\emph{(P1)}
For every $B$ and $B'$ such that $B \subseteq B' \subseteq [n]$, if $H(x_j\mid S_{\boldB})=0$, then $H(x_j\mid S_{\boldB'})=0$.
Indeed, $S_{\boldB}$ is a deterministic function of $S_{\boldB'}$, so conditioning on $S_{\boldB'}$
cannot increase the conditional entropy.

It follows that $\mathsf{Rec}(x_j)$ is \emph{upward closed} under $\subseteq$.
Consequently, the set $\alpha(x_j)$ of $\subseteq$-minimal elements of $\mathsf{Rec}(x_j)$ is an antichain:
if $\boldB_1,\boldB_2\in \alpha(x_j)$ and $\boldB_1\subsetneq \boldB_2$, then $\boldB_2$ would not be minimal.
Hence $\alpha(x_j)\in\mathcal{A}(\boldS)$.

\emph{(P2) }
By definition of $\alpha(x_j)$, we have the equivalence
\begin{equation}
\label{eq:rec-equivalence}
\boldB\in\mathsf{Rec}(x_j)
\quad\Longleftrightarrow\quad
\exists\,\boldA\in \alpha(x_j)\ \text{s.t. }\ \boldA\subseteq \boldB .
\end{equation}
The forward implication holds because any $\boldB\in\mathsf{Rec}(x_j)$ contains a minimal element of
$\mathsf{Rec}(x_j)$; the reverse implication follows from (P1).

Now for each $\alpha\in\mathcal{A}(\boldS)$ define
\[
U_\alpha \triangleq (x_j:\ j\in J_T,\ \alpha(x_j)=\alpha).
\]

We claim that $U_\alpha$ realizes exactly the \emph{intuitive first principle} of the label $\alpha$.
For any $B \in \alpha$ and for any component $x_j$ with $\alpha(x_j)=\alpha$, we have $\boldB\in\mathsf{Rec}(x_j)$ by
construction (since $\boldB$ is one of the minimal recovering groups), hence $H(x_j\mid S_{\boldB})=0$.
Therefore $H(U_\alpha\mid S_{\boldB})=0$, i.e., $U_\alpha$ is recoverable from every source group in $\alpha$.

Next, consider any strictly weaker label $\beta\prec_{\boldS}\alpha$.
By the definition of the antichain order, for each $\boldB\in\alpha$ there exists $\boldC\in\beta$ with
$\boldC\subseteq \boldB$, and strictness means that for some $\boldB^\star\in\alpha$ one can choose
$\boldC^\star\in\beta$ with $\boldC^\star\subsetneq \boldB^\star$.
Fix any component $x_j$ with $\alpha(x_j)=\alpha$ and take $\boldB^\star\in\alpha(x_j)$ corresponding to that strict
containment.
Then $\boldC^\star\notin \mathsf{Rec}(x_j)$ by minimality of $\boldB^\star$.
By Definition~\ref{def:ideal-semantics}, this implies
$I(x_j;S_{\boldC^\star})=0$.
Hence $U_\alpha$ cannot be fully recovered under the weaker label $\beta$ in the sense that at least one source
group in $\beta$ carries zero information about at least one entry of $U_\alpha$.

Finally, since the components $\{x_j\}_{j\in J_T}$ constitute $T$ (Definition~\ref{def:ideal-semantics}),
the atom value assigned to $\alpha$ is canonically determined by the collection of components whose principal
antichain equals $\alpha$, namely
\[
\Pi^T_{\boldS}(\alpha)\triangleq H(U_\alpha),
 \forall \alpha\in\mathcal{A}(\boldS),
\]
which concludes the proof.
\end{proof}

\subsection{Proof of Lemma~\ref{lemma:groundtruth-pair}}
\label{app:pair-construction}
\begin{proof}
We show that both $(\hat{\boldS},\hat{T})$ and $(\tilde{\boldS},\tilde{T})$ satisfy
Definition~\ref{def:ideal-semantics} with the same index sets
\[
J_1=\{1,4,7\}, J_2=\{2,5,8\}, J_3=\{3,6,9\}, J_T=\{1,5,9\},
\]
and hence admit canonical atoms via Lemma~\ref{lem:canonical-atom-assignment}.

\smallskip
\noindent\textbf{Step 1 (Definition~\ref{def:ideal-semantics}(i)).}

In both systems $J_T=\{1,5,9\}$. 
Hence, if $H(x_j\mid \hat{T})=0$ (resp.\ $H(x_j\mid \tilde{T})=0$), then necessarily $j\in\{1,5,9\}=J_T$.
Indeed, for any $j\notin\{1,5,9\}$, the variable $x_j$ depends on at least one latent bit that is not determined by $T$,
so $H(x_j\mid T)>0$.

\smallskip
\noindent\textbf{Step 2 (Definition~\ref{def:ideal-semantics}(ii)).}
\paragraph{For system $(\hat{\boldS},\hat{\boldT})$}
In both $\hat{S}_1=(x_1,x_4,x_7)$ and $\hat{S}_2=(x_2,x_5,x_8)$, we have that $x_1,x_4,$ and $x_7$ are mutually independent, and $x_2,x_5,$ and $x_8$ are mutually independent by construction. 
Then, $S_3=(x_3,x_6,x_9)$ has mutually independent components since each of
$x_3,x_6,x_9$ is a function of independent bits supported on disjoint
inputs.
Thus Definition~\ref{def:ideal-semantics}(i) holds for $(\hat{\boldS},\hat{\boldT})$.

\paragraph{For system $(\tilde{\boldS},\tilde{\boldT})$}
In $\tilde{S}_1=(x_1,x_4,x_7)$ we have that $x_1,x_4,$ and $x_7$ are mutually independent by construction. 
$\tilde{S}_2=(x_2,x_5,x_8)$ has mutually independent components 
because $x_8=x_7\oplus x_1\oplus x_5$ is an XOR-mask of $x_5$ by the independent Bernoulli$(1/2)$ bit $x_7\oplus x_1$.
Finally, $S_3=(x_3,x_6,x_9)$ has mutually independent components since each of $x_3,x_6,x_9$ is a function of independent bits supported on disjoint inputs ($(x_3,x_6,x_9)$ is an invertible linear transform of independent bits).
Thus Definition~\ref{def:ideal-semantics}(i) holds for $(\tilde{\boldS},\tilde{\boldT})$.

\smallskip
\noindent\textbf{Step 3 (Definition~\ref{def:ideal-semantics}(iii) and principal antichains).}

We verify $j\in J_T=\{1,5,9\}$ by identifying the minimal recovering sets.
We use a standard fact that if $u\sim\mathrm{Bern}(1/2)$ and
$u\perp v$ then $v\perp (u\oplus v)$ (this is known as XOR masking or one-time pad). Also, we have the fact that in Lemma~\ref{lem:canonical-atom-assignment}, the recoverability in~\eqref{eq:recovering} is monotone in $\boldB$.

\setcounter{paragraph}{0}
\paragraph{For system $(\hat{\boldS},\hat{T})$}

\emph{(i) Component $x_1$.}
We have $H(x_1\mid \hat{S}_1)=0$. Also $H(x_1\mid \hat{S}_2,\hat{S}_3)=0$ since $x_2\in\hat{S}_2$, $x_3\in\hat{S}_3$,
and $x_1=x_2\oplus x_3$.
Moreover, $I(x_1;\hat{S}_2)=0$ because $x_1$ is independent of $(x_2,x_5,x_8)$, and $I(x_1;\hat{S}_3)=0$ because
$x_3=x_1\oplus x_2$ is a one-time-pad masking of $x_1$ by the independent bit $x_2$ (while $x_6,x_9$ are supported on disjoint independent bits).
Hence the only minimal recovering sets are $\{1\}$ and $\{2,3\}$, so
\[
\alpha(x_1)=\bigl\{\{1\},\{2,3\}\bigr\}.
\]

\emph{(ii) Component $x_5$.}
We have $H(x_5\mid \hat{S}_2)=0$. Also $H(x_5\mid \hat{S}_1,\hat{S}_3)=0$ since $x_4\in\hat{S}_1$, $x_6\in\hat{S}_3$,
and $x_5=x_4\oplus x_6$.
Moreover, $I(x_5;\hat{S}_1)=0$ by independence, and $I(x_5;\hat{S}_3)=0$ because $x_6=x_4\oplus x_5$ is a one-time-pad masking of $x_5$ by $x_4$.
Thus
\[
\alpha(x_5)=\bigl\{\{2\},\{1,3\}\bigr\}.
\]

\emph{(iii) Component $x_9$.}
We have $H(x_9\mid \hat{S}_3)=0$, and $H(x_9\mid \hat{S}_1,\hat{S}_2)=0$ since $x_9=x_7\oplus x_8$ with
$x_7\in\hat{S}_1$ and $x_8\in\hat{S}_2$.
Moreover, $I(x_9;\hat{S}_1)=I(x_9;\hat{S}_2)=0$ since each single source provides only one addend of $x_7\oplus x_8$.
Hence
\[
\alpha(x_9)=\bigl\{\{3\},\{1,2\}\bigr\}.
\]

\paragraph{For system $(\tilde{\boldS},\tilde{T})$}
The recovery arguments are the same as in the previous system:
$H(x_1\mid \tilde{S}_1)=0$ and $H(x_1\mid \tilde{S}_2,\tilde{S}_3)=0$ via $x_1=x_2\oplus x_3$;
$H(x_5\mid \tilde{S}_2)=0$ and $H(x_5\mid \tilde{S}_1,\tilde{S}_3)=0$ via $x_5=x_4\oplus x_6$;
and $H(x_9\mid \tilde{S}_3)=0$ and $H(x_9\mid \tilde{S}_1,\tilde{S}_2)=0$ since $x_9=x_7\oplus x_8$ holds by construction.

It remains to check that no \emph{single} source reveals information about these components, so that the minimal recovering sets
stay the same.

\emph{(i)} For $x_1$, note that $\tilde{S}_2$ contains $x_8=x_7\oplus x_1\oplus x_5$, which is a one-time-pad masking of $x_1$ by the independent uniform key
$x_7\oplus x_5$ (independent of $x_1$); hence $I(x_1;\tilde{S}_2)=0$. Similarly, $\tilde{S}_3$ contains $x_3=x_1\oplus x_2$ and $x_9=x_1\oplus x_5$,
which are independent one-time-pad maskings of $x_1$ by $x_2$ and $x_5$, so $I(x_1;\tilde{S}_3)=0$.

\emph{(ii)} For $x_5$, the case $I(x_5;\tilde{S}_1)=0$ is similar to the previous system, and $I(x_5;\tilde{S}_3)=0$ still holds because
$x_6=x_4\oplus x_5$ masks $x_5$ by $x_4$ and $x_9=x_1\oplus x_5$ masks $x_5$ by $x_1$, with independent uniform keys.

\emph{(iii)} For $x_9$, the case is similar to the previous system: $\tilde{S}_1$ and $\tilde{S}_2$ each contains only one addend of $x_7\oplus x_8$ (equivalently, one masked view of $x_9$),
so $I(x_9;\tilde{S}_1)=I(x_9;\tilde{S}_2)=0$.

Therefore, in the tilde system the minimal recovering sets are the same as in the hat system, and we again obtain
\begin{align*}
\alpha(x_1)=\bigl\{\{1\},\{2,3\}\bigr\},
\alpha(x_5)=\bigl\{\{2\},\{1,3\}\bigr\},\\ \text{ and }
\alpha(x_9)=\bigl\{\{3\},\{1,2\}\bigr\}.  
\end{align*}
Thus Definition~\ref{def:ideal-semantics}(ii) holds for all $j\in J_T$ in both systems, and the principal antichains coincide.

\smallskip
\noindent\textbf{Step 4 (Coincidence of atoms).}

By Lemma~\ref{lem:canonical-atom-assignment}, the only nonzero atoms are those indexed by
$\alpha(x_1),\alpha(x_5),\alpha(x_9)$, and
\begin{align*}
\Pi^T_{\boldS}(\{\{1\}\{23\}\})=H(x_1)=1, \\
\Pi^T_{\boldS}(\{\{2\}\{13\}\})=H(x_5)=1,\\
\Pi^T_{\boldS}(\{\{3\}\{12\}\})=H(x_9)=1,  
\end{align*}

with all remaining atoms equal to $0$, in both systems. This proves Lemma~\ref{lemma:groundtruth-pair} (i).

\smallskip
\noindent\textbf{Step 5 (Different mutual informations).}

In both systems, $T=(x_1,x_5,x_9)$ is a deterministic function of $\boldS$ (since $x_1$ is in $S_1$, $x_5$ is in $S_2$,
and $x_9$ is in $S_3$), hence $H(T\mid \boldS)=0$ and $I(\boldS;T)=H(T)$.
For the hat system, $x_1,x_5,x_9$ are mutually independent, so $H(\hat{T})=3$ and $I(\hat{\boldS};\hat{T})=3$.
For the tilde system, $x_9=x_1\oplus x_5$, so $H(\tilde{T})=H(x_1,x_5)=2$ and $I(\tilde{\boldS};\tilde{T})=2$.
Thus $I(\hat{\boldS};\hat{T})\neq I(\tilde{\boldS};\tilde{T})$, proving Lemma~\ref{lemma:groundtruth-pair} (ii).
\end{proof}

\paragraph*{Proof intuition}
The three target bits $x_1,x_5,x_9$ have the same minimal recoverability patterns in both systems:
$x_1$ is recoverable from $S_1$ and from $(S_2,S_3)$ via $x_1=x_2\oplus x_3$, but not from $S_2$ or $S_3$ alone;
$x_5$ is recoverable from $S_2$ and from $(S_1,S_3)$ via $x_5=x_4\oplus x_6$, but not from $S_1$ or $S_3$ alone;
and $x_9$ is recoverable from $S_3$ and from $(S_1,S_2)$ via $x_9=x_7\oplus x_8$, but not from $S_1$ or $S_2$ alone.
Therefore $\alpha(x_1)=\{\{1\},\{23\}\}$, $\alpha(x_5)=\{\{2\},\{13\}\}$, and $\alpha(x_9)=\{\{3\},\{12\}\}$ in both cases,
which forces the same three nonzero atoms under Lemma~\ref{lem:canonical-atom-assignment}.

\subsection{Full probability tables for Fig.~\ref{fig:System23}}
\label{app:alltargets}

Tables~\ref{tab:hat-pmf} and~\ref{tab:tilde-pmf} list the full joint PMFs of
$(\hat{S}_1,\hat{S}_2,\hat{S}_3,\hat{T})$ and $(\tilde{S}_1,\tilde{S}_2,\tilde{S}_3,\tilde{T})$, respectively.
All unlisted outcomes have probability $0$.

\newpage
\clearpage
\onecolumn

% （可选）如果还很宽，可以配合横向页面：
% \usepackage{pdflscape}
% \begin{landscape}

\begingroup
\scriptsize
\renewcommand{\arraystretch}{0.95}
\setlength{\LTleft}{0pt}
\setlength{\LTright}{0pt}

\begin{center}
\scriptsize
\renewcommand{\arraystretch}{0.95}
\begin{longtable}{c c c c c}
\caption{Joint probability table of $(\hat{S}_1,\hat{S}_2,\hat{S}_3,\hat{T})$ in Fig.~\ref{fig:System23}.}\label{tab:hat-pmf}\\
\hline
$\hat{S}_1$ & $\hat{S}_2$ & $\hat{S}_3$ & $\hat{T}$ & $\Pr$ \\
$(x_1,x_4,x_7)$ & $(x_2,x_5,x_8)$ & $(x_3,x_6,x_9)$ & $(x_1,x_5,x_9)$ & \\
\hline
\endfirsthead
\hline
$\hat{S}_1$ & $\hat{S}_2$ & $\hat{S}_3$ & $\hat{T}$ & $\Pr$ \\
$(x_1,x_4,x_7)$ & $(x_2,x_5,x_8)$ & $(x_3,x_6,x_9)$ & $(x_1,x_5,x_9)$ & \\
\hline
\endhead
\texttt{000} & \texttt{000} & \texttt{000} & \texttt{000} & $2^{-6}$ \\
\texttt{000} & \texttt{001} & \texttt{001} & \texttt{001} & $2^{-6}$ \\
\texttt{000} & \texttt{010} & \texttt{010} & \texttt{010} & $2^{-6}$ \\
\texttt{000} & \texttt{011} & \texttt{011} & \texttt{011} & $2^{-6}$ \\
\texttt{000} & \texttt{100} & \texttt{100} & \texttt{000} & $2^{-6}$ \\
\texttt{000} & \texttt{101} & \texttt{101} & \texttt{001} & $2^{-6}$ \\
\texttt{000} & \texttt{110} & \texttt{110} & \texttt{010} & $2^{-6}$ \\
\texttt{000} & \texttt{111} & \texttt{111} & \texttt{011} & $2^{-6}$ \\
\texttt{001} & \texttt{000} & \texttt{001} & \texttt{001} & $2^{-6}$ \\
\texttt{001} & \texttt{001} & \texttt{000} & \texttt{000} & $2^{-6}$ \\
\texttt{001} & \texttt{010} & \texttt{011} & \texttt{011} & $2^{-6}$ \\
\texttt{001} & \texttt{011} & \texttt{010} & \texttt{010} & $2^{-6}$ \\
\texttt{001} & \texttt{100} & \texttt{101} & \texttt{001} & $2^{-6}$ \\
\texttt{001} & \texttt{101} & \texttt{100} & \texttt{000} & $2^{-6}$ \\
\texttt{001} & \texttt{110} & \texttt{111} & \texttt{011} & $2^{-6}$ \\
\texttt{001} & \texttt{111} & \texttt{110} & \texttt{010} & $2^{-6}$ \\
\texttt{010} & \texttt{000} & \texttt{010} & \texttt{000} & $2^{-6}$ \\
\texttt{010} & \texttt{001} & \texttt{011} & \texttt{001} & $2^{-6}$ \\
\texttt{010} & \texttt{010} & \texttt{000} & \texttt{010} & $2^{-6}$ \\
\texttt{010} & \texttt{011} & \texttt{001} & \texttt{011} & $2^{-6}$ \\
\texttt{010} & \texttt{100} & \texttt{110} & \texttt{000} & $2^{-6}$ \\
\texttt{010} & \texttt{101} & \texttt{111} & \texttt{001} & $2^{-6}$ \\
\texttt{010} & \texttt{110} & \texttt{100} & \texttt{010} & $2^{-6}$ \\
\texttt{010} & \texttt{111} & \texttt{101} & \texttt{011} & $2^{-6}$ \\
\texttt{011} & \texttt{000} & \texttt{011} & \texttt{001} & $2^{-6}$ \\
\texttt{011} & \texttt{001} & \texttt{010} & \texttt{000} & $2^{-6}$ \\
\texttt{011} & \texttt{010} & \texttt{001} & \texttt{011} & $2^{-6}$ \\
\texttt{011} & \texttt{011} & \texttt{000} & \texttt{010} & $2^{-6}$ \\
\texttt{011} & \texttt{100} & \texttt{111} & \texttt{001} & $2^{-6}$ \\
\texttt{011} & \texttt{101} & \texttt{110} & \texttt{000} & $2^{-6}$ \\
\texttt{011} & \texttt{110} & \texttt{101} & \texttt{011} & $2^{-6}$ \\
\texttt{011} & \texttt{111} & \texttt{100} & \texttt{010} & $2^{-6}$ \\
\texttt{100} & \texttt{000} & \texttt{100} & \texttt{100} & $2^{-6}$ \\
\texttt{100} & \texttt{001} & \texttt{101} & \texttt{101} & $2^{-6}$ \\
\texttt{100} & \texttt{010} & \texttt{110} & \texttt{110} & $2^{-6}$ \\
\texttt{100} & \texttt{011} & \texttt{111} & \texttt{111} & $2^{-6}$ \\
\texttt{100} & \texttt{100} & \texttt{000} & \texttt{100} & $2^{-6}$ \\
\texttt{100} & \texttt{101} & \texttt{001} & \texttt{101} & $2^{-6}$ \\
\texttt{100} & \texttt{110} & \texttt{010} & \texttt{110} & $2^{-6}$ \\
\texttt{100} & \texttt{111} & \texttt{011} & \texttt{111} & $2^{-6}$ \\
\texttt{101} & \texttt{000} & \texttt{101} & \texttt{101} & $2^{-6}$ \\
\texttt{101} & \texttt{001} & \texttt{100} & \texttt{100} & $2^{-6}$ \\
\texttt{101} & \texttt{010} & \texttt{111} & \texttt{111} & $2^{-6}$ \\
\texttt{101} & \texttt{011} & \texttt{110} & \texttt{110} & $2^{-6}$ \\
\texttt{101} & \texttt{100} & \texttt{001} & \texttt{101} & $2^{-6}$ \\
\texttt{101} & \texttt{101} & \texttt{000} & \texttt{100} & $2^{-6}$ \\
\texttt{101} & \texttt{110} & \texttt{011} & \texttt{111} & $2^{-6}$ \\
\texttt{101} & \texttt{111} & \texttt{010} & \texttt{110} & $2^{-6}$ \\
\texttt{110} & \texttt{000} & \texttt{110} & \texttt{100} & $2^{-6}$ \\
\texttt{110} & \texttt{001} & \texttt{111} & \texttt{101} & $2^{-6}$ \\
\texttt{110} & \texttt{010} & \texttt{100} & \texttt{110} & $2^{-6}$ \\
\texttt{110} & \texttt{011} & \texttt{101} & \texttt{111} & $2^{-6}$ \\
\texttt{110} & \texttt{100} & \texttt{010} & \texttt{100} & $2^{-6}$ \\
\texttt{110} & \texttt{101} & \texttt{011} & \texttt{101} & $2^{-6}$ \\
\texttt{110} & \texttt{110} & \texttt{000} & \texttt{110} & $2^{-6}$ \\
\texttt{110} & \texttt{111} & \texttt{001} & \texttt{111} & $2^{-6}$ \\
\texttt{111} & \texttt{000} & \texttt{111} & \texttt{101} & $2^{-6}$ \\
\texttt{111} & \texttt{001} & \texttt{110} & \texttt{100} & $2^{-6}$ \\
\texttt{111} & \texttt{010} & \texttt{101} & \texttt{111} & $2^{-6}$ \\
\texttt{111} & \texttt{011} & \texttt{100} & \texttt{110} & $2^{-6}$ \\
\texttt{111} & \texttt{100} & \texttt{011} & \texttt{101} & $2^{-6}$ \\
\texttt{111} & \texttt{101} & \texttt{010} & \texttt{100} & $2^{-6}$ \\
\texttt{111} & \texttt{110} & \texttt{001} & \texttt{111} & $2^{-6}$ \\
\texttt{111} & \texttt{111} & \texttt{000} & \texttt{110} & $2^{-6}$ \\
\hline
\end{longtable}
\end{center}

\newpage
\begin{center}
\scriptsize
\renewcommand{\arraystretch}{0.95}
\begin{longtable}{c c c c c}
\caption{Joint probability table of $(\tilde{S}_1,\tilde{S}_2,\tilde{S}_3,\tilde{T})$ in Fig.~\ref{fig:System23}.}\label{tab:tilde-pmf}\\
\hline
$\tilde{S}_1$ & $\tilde{S}_2$ & $\tilde{S}_3$ & $\tilde{T}$ & $\Pr$ \\
$(x_1,x_4,x_7)$ & $(x_2,x_5,x_8)$ & $(x_3,x_6,x_9)$ & $(x_1,x_5,x_9)$ & \\
\hline
\endfirsthead
\hline
$\tilde{S}_1$ & $\tilde{S}_2$ & $\tilde{S}_3$ & $\tilde{T}$ & $\Pr$ \\
$(x_1,x_4,x_7)$ & $(x_2,x_5,x_8)$ & $(x_3,x_6,x_9)$ & $(x_1,x_5,x_9)$ & \\
\hline
\endhead
\texttt{000} & \texttt{000} & \texttt{000} & \texttt{000} & $2^{-5}$ \\
\texttt{000} & \texttt{011} & \texttt{011} & \texttt{011} & $2^{-5}$ \\
\texttt{000} & \texttt{100} & \texttt{100} & \texttt{000} & $2^{-5}$ \\
\texttt{000} & \texttt{111} & \texttt{111} & \texttt{011} & $2^{-5}$ \\
\texttt{001} & \texttt{001} & \texttt{000} & \texttt{000} & $2^{-5}$ \\
\texttt{001} & \texttt{010} & \texttt{011} & \texttt{011} & $2^{-5}$ \\
\texttt{001} & \texttt{101} & \texttt{100} & \texttt{000} & $2^{-5}$ \\
\texttt{001} & \texttt{110} & \texttt{111} & \texttt{011} & $2^{-5}$ \\
\texttt{010} & \texttt{000} & \texttt{010} & \texttt{000} & $2^{-5}$ \\
\texttt{010} & \texttt{011} & \texttt{001} & \texttt{011} & $2^{-5}$ \\
\texttt{010} & \texttt{100} & \texttt{110} & \texttt{000} & $2^{-5}$ \\
\texttt{010} & \texttt{111} & \texttt{101} & \texttt{011} & $2^{-5}$ \\
\texttt{011} & \texttt{001} & \texttt{010} & \texttt{000} & $2^{-5}$ \\
\texttt{011} & \texttt{010} & \texttt{001} & \texttt{011} & $2^{-5}$ \\
\texttt{011} & \texttt{101} & \texttt{110} & \texttt{000} & $2^{-5}$ \\
\texttt{011} & \texttt{110} & \texttt{101} & \texttt{011} & $2^{-5}$ \\
\texttt{100} & \texttt{000} & \texttt{100} & \texttt{100} & $2^{-5}$ \\
\texttt{100} & \texttt{011} & \texttt{111} & \texttt{111} & $2^{-5}$ \\
\texttt{100} & \texttt{100} & \texttt{000} & \texttt{100} & $2^{-5}$ \\
\texttt{100} & \texttt{111} & \texttt{011} & \texttt{111} & $2^{-5}$ \\
\texttt{101} & \texttt{001} & \texttt{100} & \texttt{100} & $2^{-5}$ \\
\texttt{101} & \texttt{010} & \texttt{111} & \texttt{111} & $2^{-5}$ \\
\texttt{101} & \texttt{101} & \texttt{000} & \texttt{100} & $2^{-5}$ \\
\texttt{101} & \texttt{110} & \texttt{011} & \texttt{111} & $2^{-5}$ \\
\texttt{110} & \texttt{000} & \texttt{110} & \texttt{100} & $2^{-5}$ \\
\texttt{110} & \texttt{011} & \texttt{101} & \texttt{111} & $2^{-5}$ \\
\texttt{110} & \texttt{100} & \texttt{010} & \texttt{100} & $2^{-5}$ \\
\texttt{110} & \texttt{111} & \texttt{001} & \texttt{111} & $2^{-5}$ \\
\texttt{111} & \texttt{001} & \texttt{110} & \texttt{100} & $2^{-5}$ \\
\texttt{111} & \texttt{010} & \texttt{101} & \texttt{111} & $2^{-5}$ \\
\texttt{111} & \texttt{101} & \texttt{010} & \texttt{100} & $2^{-5}$ \\
\texttt{111} & \texttt{110} & \texttt{001} & \texttt{111} & $2^{-5}$ \\
\hline
\end{longtable}
\end{center}

\endgroup

\clearpage
\twocolumn

}

\end{document}